\documentclass[final,3p,times,twocolumn]{elsarticle}

\usepackage[T1]{fontenc}

\usepackage{mathtools}

\usepackage{amsmath,amssymb,amsfonts,amstext,graphics,graphicx,subfigure,dcolumn,bm}
\usepackage{setspace}

\usepackage[colorlinks=false]{hyperref}
\usepackage{xcolor,cancel}
\usepackage{color}

\def\bq{\begin{equation}}
\def\eq{\end{equation}}
\def\bqy{\begin{eqnarray}}
\def\eqy{\end{eqnarray}}

\def\bal#1\eal{\begin{align}#1\end{align}}

\newtheorem{lemma}{Lemma}

\usepackage{amssymb}

\def\bfJ{\mathbf{J}}
\def\bfP{\mathbf{P}}

\def\bfm{\mathbf{m}}

\def\bfx{\mathbf{x}}

\def\bfv{\mathbf{v}}
\def\bfu{\mathbf{u}}
\def\bfp{\mathbf{p}}

\def\bfp{\mathbf{p}}

\def\bfxi{\boldsymbol{\xi}}

\def\calb{\mathcal{B}}
\def\calc{\mathcal{C}}

\def\calf{\mathcal{F}}

\def\call{\mathcal{L}}

\def\R{\mathbb{R}}

\def\bq{\begin{equation}}
\def\eq{\end{equation}}
\def\bqy{\begin{eqnarray}}
\def\eqy{\end{eqnarray}}

\def\bal#1\eal{\begin{align}#1\end{align}}

\def\quadd{\quad\quad}

\def\4tensLa{\bar{\bar{\Lambda}}}

\def\tensLa{\bar{\Lambda}}
\def\tensI{\bar{I}}
\def\tensD{\bar{D}}
\def\tenska{\bar{\kappa}}
\def\tensJ{\bar{J}}

\def\al{\alpha}
\def\be{\beta}

\def\de{\delta}

\def\ep{\epsilon}

\def\et{\eta}
\def\ga{\gamma}
\def\Ga{\Gamma}

\def\la{\lambda}
\def\La{\Lambda}

\def\Om{\Omega}

\def\si{\sigma}
\def\Si{\Sigma}

\def\ze{\zeta}

\def\p{\partial}

\def\nn{\nonumber}

\newtheorem{proof}{Proof}

\journal{ }

\begin{document}

\begin{frontmatter}


\title{Thermodynamically consistent Cahn-Hilliard-Navier-Stokes  equations using the metriplectic
dynamics formalism}



\author[inst1,inst2]{Azeddine Zaidni}
\affiliation[inst1]{organization={Mohammed VI Polytechnic University, College of Computing lab},
            addressline={Lot 660, Hay Moulay Rachid }, 
            city={Ben Guerir},
            postcode={43150}, 
            country={Morocco}}

\author[inst2]{Philip J Morrison}
\affiliation[inst2]{organization={Department of Physics and Institute for Fusion Studies, University of Texas at Austin},
            city={Austin},
            postcode={78712}, 
            state={TX},
            country={USA}}
\author[inst3]{Saad Benjelloun}
\affiliation[inst3]{organization={Makhbar Mathematical Sciences Research Institute},
            city={Casablanca},
            country={Morocco}}

\begin{abstract}
Cahn-Hilliard-Navier-Stokes (CHNS) systems describes flows with two-phases, e.g., a liquid with bubbles. Obtaining constitutive relations for general dissipative processes for such a systems, which are thermodynamically consistent, can be a challenge. We show how the metriplectic 4-bracket formalism \citep{pjmU23}  achieves this in a straightforward, in fact algorithmic,  manner. First, from the noncanonical Hamiltonian formulation for the ideal part of a CHNS system we obtain an appropriate Casimir to serve as the entropy in the metriplectic  formalism  that describes the dissipation (e.g. viscosity, heat conductivity and diffusion effects). General thermodynamics with the concentration variable and its thermodynamics conjugate,
the chemical potential, are included. 
Having expressions for the Hamiltonian (energy), entropy, and Poisson bracket, we describe a procedure for obtaining a metriplectic 4-bracket that describes thermodynamically consistent dissipative effects.  The 4-bracket formalism leads naturally to a general CHNS system that allows for anisotropic surface energy effects.  This  general CHNS system  reduces to cases in the literature, to which we can compare. 
\end{abstract}




\end{frontmatter}


\tableofcontents

\newpage
\section{Introduction}

The well-known Navier-Stokes equations govern the motion of a single-phase fluid. However, in the case of two-phase fluids,  chemical reactions,  changes of  phase, and  migration between substances of phases become significant and cannot be disregarded. J.\ W.\  Cahn and J.~E.\ Hilliard were the first to formulate the mathematical equations that describe  phase separation in a  such a binary fluid  \citep{cahn1958free}. Here we investigate generalizations that  combine   the Cahn-Hilliard equation with equations that describe the dynamics of fluid flow, referred to as  Cahn-Hilliard-Navier-Stokes (CHNS) systems. CHNS systems aim to describe the hydrodynamic properties of a mixture of two phases such as bubbles in a  liquid.  To narrow down the already broad scope,  we assume that the two fluids share the same velocity field, yet we allow for both extended thermodynamics and diffusive interfaces between the two phases.

A substantial hurdle in developing CHNS type systems, systems with a variety of constitutive relations,  is to ensure thermodynamic consistency, i.e., adherence to the first law of thermodynamics, which in this context is to produce a set of dynamical equations that conserve  energy, and the second law which in this context means the  dynamical production of entropy,  ensuring the relaxation  (asymptotic stability) to thermodynamic equilibrium.  Here we propose an algorithm for constructing such systems, an algorithm that produces a large set of CHNS systems.  

The algorithm has four steps: i) Select a set of  dynamical variables. For a CHNS system these will be 
$\psi\coloneqq\{\bfm=\rho\bfv, \rho, \tilde{c}=\rho  c,  \si=\rho   s\}$, which  are  the  momentum density, mass density, volume concentration of one of the constituents, and entropy density, respectively. ii) The next step is to  select  energy and entropy functionals,  $H$ and $S$, dependent on the dynamical variables.  The choice of these functionals is based on the physics of  the phenomena one wishes to describe.  iii) The third step of the algorithm is to obtain the noncanonical Poisson bracket  \citep[see][]{pjm98} of  the ideal (nondissipative) part of the  theory that has the chosen entropy as a Casimir invariant.  Since the work of  \citet{pjmG80}, Poisson brackets for a great many systems, including fluid and magnetofluid systems, have been found  \citep[e.g.][]{pjm82,pjmT00,hamdi,pjmDL16,zaidni1}.  Thus, this step may be immediate.  Alternatively,  it  may be achieved by a coordinate change from a known Hamiltonian theory in order to align with the chosen entropy functional.  In either case,  we obtain at this stage a noncanonical Hamiltonian system.  iv) The final step is to construct a metriplectic 4-bracket as described in \citet{pjmU23}.  Although there are standard metriplectic 4-bracket constructions, there is freedom at this last step to describe a variety of types of dissipation.  However, a natural choice follows upon consideration of the form of an early  metriplectic bracket \citep{pjm84b}.  Given $H$, $S$, and the 4-bracket, the dynamical  system with  thermodynamically consistent dissipation is produced.

We apply the algorithm to two cases. First, in \S~\ref{sec:GNS}, we consider a system where the fluid  thermodynamics is extended by allowing the internal energy to depend on a concentration variable,  with the chemical potential being  its thermodynamic dual.   Because Gibbs introduced the notion of chemical potential,   we refer to the Hamiltonian version of this fluid systems as the Gibbs-Euler (GE) system and the dissipative version    as the  Gibbs-Navier-Stokes  (GNS) system. It is a  
thermodynamically consistent version of the compressible Navier-Stokes equations with the inclusion of this  concentration variable for describing a second phase of the fluid.   The GNS system generalizes  the early work of \citet{eckart1,eckart2} and the  treatment in \citep{de2013non}; it allows for all possible thermodynamics fluxes.  Next, in \S~\ref{sec:DICHNS},  a general form  of  CHNS  system  is produced, a form that models surface tension effects and allows for diffuse interfaces.  Our work is motivated in large part by  the substantial works of  \citet{anderson2000phase} and  \citet{Guo}, which we generalize by obtaining a class of systems that includes theirs as special cases. There is a huge literature on this topic  and these papers contain  many important references to previous work.   \citep[Also, see][for a recent review.]{Eikelder23} 

The GNS system of  \S~\ref{sec:GNS} serves as a  straightforward example of our algorithm.  In \S\S~\ref{ssec:SIDescription} we describe the set of dynamical variables, properties of the system, the energy and entropy functionals $H$ and $S$.  This amounts to the first and second steps of the  algorithm.  Then in \S\S~\ref{ssec:SIPbkt} the Hamiltonian formulation of the dissipation free part of the system is presented.    This is the third step of the algorithm where the Poisson bracket is obtained, after a brief review of the noncanonical Hamiltonian formalism.  Given the early work of \citet{pjmG80} and the classification  of  extensions in \citep{pjmT00}, this step is immediate.  Based on the early and recent works of \citet{pjm84b} and \cite{pjmU23} the fourth step of the algorithm is also immediate. In   \S\S~\ref{ssec:SI4bkt} we first review the metriplectic 4-bracket formalism and present the realization that applies for the GNS system.  Thus, the thermodynamically consistent GNS  system  is determined. In \S\S~\ref{ssec:SI2bkt} we obtain the metriplectic 2-bracket  equations of motion, and the determined fluxes and affinities, making connection to standard irreversible thermodynamics. 
Using the results of  \S~\ref{sec:GNS}, we proceed in \S~\ref{sec:DICHNS} to obtain the main result of the paper, our general CHNS system that can describe diffuse interface effects. The first and second steps of our algorithm are taken  in  \S\S~\ref{ssec:chnsHS}, while the third step, obtaining the correct  Poisson bracket,  is undertaken in \S\S~\ref{ssec:chnsPB}.  In order to complete this step, one must find the Poisson bracket for which the entropy of the second step is a Casimir invariant, which we find can be achieved by a simple coordinate transformation.  The fourth step of the algorithm is taken  in  \S\S~\ref{ssec:chns4MB}.  Here a   choice of metriplectic 4-bracket gives a general class of thermodynamically consistent CHNS systems, a class that contains previous results as special cases.   The formalism also shows how one can transform to a simple entropy variable at the expense of a more complicated internal energy. As in \S~\ref{sec:GNS}, in \S\S~\ref{ssec:chns2MB} we reduce to  the metriplectic 2-bracket. 
Finally in \S~\ref{sec:conclu} we briefly summarize and make a few comments about ongoing and future work.

\section{Metriplectic framework and the Gibbs-Navier-Stokes system}
\label{sec:GNS}

In this section we describe general features of the metriplectic framework in the context of the  GNS system, a generalization of the Navier-Stokes equations that includes the dual thermodynamical variables  of  concentration and chemical potential. 

\subsection{Description of the Gibbs-Navier-Stokes system}
\label{ssec:SIDescription}

The GNS for 2 phase flow proceeds on familiar ground  \citep{eckart1,eckart2,de2013non}.  It amounts to the single phase thermodynamic  Navier-Stokes  system or as it is sometimes called the Fourier Navier-Stokes  system with the dispersed   
phase described by the addition of a concentration variable, $c$, giving the set of dynamical variables $\psi=\{ \bfv, \rho, c, s \}$.   Here we review global aspects of this known system, before showing how it emerges from the metriplectic formalism.

We suppose the mixture of two phases are contained in a  volume  $\Omega$, and we consider the following global quantities and their evolution:
\bal
    {M} &= \int_{\Omega} \rho\,,\ \ 
    \dot{M} = 0,
    \label{glbM}
\\
       \mathbf{P} &= \int_{\Omega} \rho\,  \mathbf{v}\,, \ \ 
        \dot{\bfP} = -\int_{\partial\Omega}  \bar{J}_{\mathbf{m}}\cdot\mathbf{n},
\\
    H &=\!  \int_{\Omega} \frac{\rho}{2}\, |\mathbf{v}|^2   + \rho\, u(\rho,s,c)\,,
    \  \ 
      \dot{H}= \!-\int_{\partial\Omega} {\bfJ}_{e} \cdot\mathbf{n}\,,
         \label{hamH}
         \\
           {C} &=\!\int_{\Omega} \rho\,  c\,, \ \ 
   \dot{C} = \!- \int_{\partial \Omega}  {\bfJ}_c\cdot\mathbf{n},
\\
     {S} &= \int_{\Omega} \rho \, s\,, \ \ 
  \dot{S}= -\int_{\partial\Omega}  {\bfJ}_{s}\cdot\mathbf{n} + \int_{\Omega} \Dot{s}^{prod}\,.
    \label{entS}
\eal
Here  $\rho$ is the density of the mixture, $\mathbf{v}$ is the mass-averaged velocity of the mixture, $s$ is the specific entropy, and the phase variable $c$ is the specific concentration {(dimensionless mass concentration)}  that determines how much of the dispersed  phase  of the mixture is present at a point  $\bfx \in \Omega\subset\R^3$.  The variable $\tilde{c}=\rho c$ is  the mass density of the dispersed phase. The  functionals ${M}$, ${P}$, $H$ and ${S}$ are the total mass, momentum, energy, and entropy of the mixture,  respectively, while ${C}$ is the total {mass of one of the constituents}.  For convenience we will omit the incremental volume element for integrations over $\Omega$, i.e., $\int_\Om= \int_\Om d^3x$ and we used an over dot to mean the total derivative $d/dt$.  The local thermodynamics  of the mixture is described by $u(\rho,s,c)$,  the internal energy per unit mass.  For convenience the gravitational force  is not considered, although its inclusion is straightforward.

Quantities in the time derivatives of the basic functionals are as follows:   $\mathbf{n}$ is  the unit outward normal vector of the boundary $\partial\Omega$, ${\bfJ}_c$ is the phase field flux, which depends on gradient of the chemical potential,  
${\tensJ}_{\mathbf{m}}$ is the stress tensor -- surface forces -- due to pressure and viscosity, ${\bfJ}_e$ the energy flux that contains the rate of work done by the surface forces (external energy), the rate of heat transfer and the rate of diffusivity in phase field (internal energy), ${\bfJ}_s$ is the net entropy flux through the boundary, and  $\Dot{s}^{prod}$ is the local rate of entropy production. The second law of thermodynamics is expressed by the requirement that  $\Dot{s}^{prod}$ is non-negative. 

For the GNS system the fluxes are given by 
\bal
    {\bfJ}_c &= - \tensD\cdot \nabla \mu\,,
\\
     \bar{J}_{\mathbf{m}} &=  p\,{\tensI} - \4tensLa:\nabla \mathbf{v}\,,
\\
     \bfJ_e &= -\mathbf{v}\cdot\4tensLa:\nabla\mathbf{v} - \tenska\cdot \nabla T -\mu \tensD\cdot\nabla\mu\,,  
    \label{eneryflux}
\\
     {\bfJ}_s &= - \frac{\tenska}{T}\cdot\nabla T\,,
\\
     \Dot{s}^{prod} &= \frac{1}{T}\Big[\nabla\mathbf{v}:\4tensLa:\nabla\bfv+ \frac{1}{T}\, \nabla T\cdot\tenska \cdot\nabla T 
    \nn\\
    &  \hspace{2 cm} +  \nabla\mu\cdot \tensD\cdot\nabla\mu\Big] \geq 0\,,
     \label{entropy1}
\eal
where  $p$ is the pressure, $T$ is the temperature,  $\tensI$ is the unit tensor, $\tenska$ is the  thermal conductivity tensor, $\tensD$ is the  diffusion tensor, which along with $\tenska$ is assumed to be a symmetric and positive definite  2-tensor,   and $\mu$ is the chemical potential.  We allow  the possibility that  phenomenological quantities such as  $\tenska$ and $\tensD$ can  depend on the dynamical variables. Here, 
$\4tensLa$ is the viscosity 4-tenor, the usual rank 4 isotropic Cartesian tensor  given by
\bq
\La_{ijkl}=\et\left(\de_{il} \de_{jk}+ \de_{jl}\de_{ik} -\frac2{3} \de_{ij}\de_{kl} \right) + \ze\, \de_{ij}\de_{kl}\,,
\label{viscous}
\eq
with viscosity coefficients $\et$ and $\ze$ and $i,j,k$ and $l$ taking on values 1,2,3.  Note, we use boldface as in the fluxes ${\bfJ}_c, {\bfJ}_e$, and 
${\bfJ}_s$  to denote vectors, an over bar as in ${\tensJ}_{\mathbf{m}}$ to denote  rank-2  tensors, and a double over bar as in  $\4tensLa$ to denote rank-4 tensors.   A single ``$\,\cdot\,$" is used for neighboring contractions, e.g.,  $(\tensD\cdot \nabla \mu)_i=D_{ij} \p_j \mu$ and we use the double dot convention, e.g., for the stress tensor   
$(\4tensLa:\!\nabla \bfv)_{ij}=\La_{ijkl} \p_k v_l$, where repeated indices are summed over.

The volume density variables are $\psi=(\rho,\mathbf{m}\coloneqq\rho \mathbf{v}, \sigma \coloneqq \rho s, \Tilde{c} \coloneqq \rho c)$, where  $\mathbf{m}$ is the momentum density, $\sigma$ is entropy per unit volume and $\Tilde{c}$ is the concentration per unit volume. The local energy per unit volume is  given by 
\bq
e =  \frac{\, |\mathbf{m}|^2}{2\rho} + \rho u(\rho,s,c)\,.
\label{e}
\eq
From the specific  internal energy,   $u(\rho,s,c)$,  we have the thermodynamic relations
\begin{equation}
\mathrm{d}u = T \mathrm{d}s + \frac{p}{\rho^2} \mathrm{d}\rho + \mu \mathrm{d} c\,,
 \label{TH1}
\end{equation}
where
\bq
 T =  \frac{\partial u }{\partial s}\,,  \qquad p =  \rho^2 \frac{\partial u}{\partial \rho}\,,
 \quadd\mu =  \frac{\partial u}{\partial c} \,.
 \label{TH2a} 
\eq

Given the content of this section, we have established the first step of our algorithm for the GNS system, the determination of the dynamical variables $\psi=\{ \bfm=\rho\bfv, \rho, \tilde{c}=\rho  c, \si=\rho   s\}$ or alternatively the set $(\mathbf{v}, \rho,   c,s)$, and the second  step of our algorithm by making the choices of Hamiltonian $H$ of \eqref{hamH}  and entropy $S$ of \eqref{entS}.  In the next section, \S\S~\eqref{ssec:SIPbkt}, we proceed to the third step of the algorithm by obtaining the Hamiltonian structure for this system.  This system without dissipation is the GE system.

\subsection{Noncanonical Poisson bracket of the Gibbs-Euler system}
\label{ssec:SIPbkt}

Given that  the mixture is assumed to be confined in the domain $\Omega$, the Eulerian scalars (volume forms) $(\rho, \Tilde{c}, \sigma)$
are functions from space-time $\Omega \mapsto \mathbb{R} \rightarrow \mathbb{R}$, while the vector field $\bfm$ maps 
$\Omega \times \mathbb{R} \mapsto T\Omega$, where $T\Omega$ stands for the tangent bundle of the manifold $\Omega$.  {We will forgo formal geometric considerations and suppose our infinite-dimensional phase space has coordinates $\psi=(\bfm,\rho, \Tilde{c}, \sigma)$ and observables are  functionals  that map  $\psi\mapsto \R$ at each fixed time.  We will denote the space of such functionals by $\calb$.  Then a Poisson  bracket is  a bilinear operator $\calb\times \calb  \mapsto \mathbb{R}$ 
that fulfills the Leibniz rule and is a realization of a Lie algebra \citep[see e.g.][chap.\ 14]{Sudarshan}.  The Leibniz rule follows from that for the variational or functional derivative of $F\in \calb$, defined by
$$
\delta F[\psi; \delta \psi) = \lim_{\epsilon \rightarrow 0} \frac{F(\psi+\epsilon\delta \psi) - F(\psi)}{\epsilon}
=\int_{\Om} \frac{\de F}{\de \psi}\de \psi,
$$
where $\de F/\de\psi$ is the functional derivative.  This expression  can be viewed as the directional derivative of a functional $F$ at $\psi$ in the direction $\delta \psi$ 
\citep[see, e.g.,][for a formal review of these notions]{pjm98}.

The appropriate Poisson bracket, defined on two functionals $F,G\in \calb$,  for the GE system is the following:
\bal
\{F,G\} &= -  \int_{\Omega} \mathbf{m}\cdot \left[  F_{\bfm} \cdot \nabla G_{\bfm} - G_{\bfm}\cdot \nabla F_{\bfm}\right]
 \nn\\
&+ \rho\left[F_{\bfm}\cdot \nabla G_\rho  -G_{\bfm}\cdot \nabla F_\rho  \right]
 \nn\\
&+ \sigma \left[F_{\bfm}\cdot \nabla G_\si - G_{\bfm}\cdot\nabla  F_\si\right] 
 \nn\\
&+ \Tilde{c} \left[F_{\bfm}\cdot \nabla  G_{\Tilde{c}} - G_{\bfm}\cdot\nabla    F_{\Tilde{c}} \right] \,,
\label{PB1}
\eal
where we compactified our notation by defining $F_{\bfm}\!\!\coloneqq\de F/\de \bfm$, $F_\rho\!\!\coloneqq\de F/\de \rho$,   etc., the functional derivatives with respect to the various coordinates $\psi$. That this is the appropriate Poisson bracket is immediate;  it  is the Lie-Poisson bracket  originally given by \citet{pjmG80} with the addition of the last line of \eqref{PB1} involving the concentration, another volume density variable $\tilde{c}$.   Adding such a dynamical variable is common place in the fluid modeling of plasmas over the last decades and fits within the general theory for extension given by  \citet{pjmT00}.  By construction we have a Poisson bracket that is a bilinear, antisymmetric, and either by the extension theory or a relatively easy direct calculation using the techniques of \citet{pjm82} it can be shown to satisfy the Jacobi identity, i.e.,  
\bq
 \{\{F, G\}, H\} + \{\{H, F\}, G\} + \{\{G, H\}, F\} = 0\,, 
 \eq
for all $F, G, H\in\calb$.  The Leibniz property, which is required for the Poisson bracket to generate a vector  field,  is built into the definition of functional derivative.

Upon inserting any functional of $\psi$, say an observable $o$,  into the Poisson bracket  its evolution is determined by 
\bq
    \partial_t o= \{o, {H}\}\,, 
    \label{fH}
\eq
where the Hamiltonian functional is the total energy of the system, where we rewrite \eqref{hamH} as follows:
\bq
 {H}[\rho,\mathbf{m},\sigma,\Tilde{c}]= \int_{\Omega} \!e \, = \int_{\Omega}  \frac{\,|\mathbf{m}|^2}{2\rho} + \rho u\left(\rho,\frac{\sigma}{\rho},\frac{\Tilde{c}}{\rho}\right).
 \label{Hamc}
\eq
In \eqref{fH} and henceforth we use the shorthand $\p_t =\p/ \p t$. Using the following functional derivatives:
\bal
 H_\rho &= -{|\mathbf{v}|^2}/{2}+ u +  {p}/{\rho} -sT-c\mu\,, \quad H_\bfm = \mathbf{v}\,,
  \nn\\
  H_\si&= T\,, \quad  H_{\Tilde{c}} = \mu\,,
    \label{var1}
\eal
the bracket form of \ref{fH} gives the  ideal two-phase flow system
\bal
  \p_t \mathbf{v} &=\{\bfv,H\}= -\bfv\cdot\nabla \bfv -\nabla p/\rho\,,
   \label{pvt}
\\
  \p_t\rho &=\{\rho,H\}= -\bfv\cdot\nabla\rho   -\rho\, \nabla\cdot\mathbf{v}\,,
      \label{prt}
      \\
    \p_t \Tilde{c}  &=\{\tilde{c},H\}= -\bfv\cdot\nabla\tilde{c} -\Tilde{c}\,\nabla\cdot\mathbf{v}\,,
      \label{pct}
    \\
\p_t\si&=\{\si ,H\}= -\bfv\cdot\nabla\si  -\sigma\,\nabla\cdot\mathbf{v}\,.
      \label{pst}
\eal 
Here we have dropped surface terms arising from integration by parts and have used $\de \rho(\bfx)/\de \rho(\bfx') =\de(\bfx-\bfx')$. 
Equations \eqref{pvt}--\eqref{pct} can also be written easily using ,  e.g., ${\mathrm{D} \rho}/{\mathrm{D}t}\coloneqq {\partial \rho}/{\partial t} + \mathbf{v}\cdot\nabla \rho$. These equations comprise  the  GE system.

\bigskip

Casimir invariants are special functionals $\mathfrak{C}$ that satisfy 
\bq
\{F, \mathfrak{C}\}= 0 \qquad \forall F\in\calb\,,
\eq
{and thus are constants of motion for any Hamiltonian.} From \eqref{PB1} we obtain the following equations that a  Casimir functional $\mathfrak{C}$ must satisfy:
\bq
\nabla\cdot(\rho \,  \mathfrak{C}_{\bfm})=\nabla\cdot(\si \,  \mathfrak{C}_{\bfm})=\nabla\cdot(\tilde{c} \,   \mathfrak{C}_{\bfm})=0\,,
\eq
 and 
\bq
m_j \nabla \mathfrak{C}_{m_j} +\p_j(\bfm \, \mathfrak{C}_{m_j}) + \rho\, \nabla \mathfrak{C}_\rho + \si\, \nabla \mathfrak{C}_\si + \tilde{c} \, \nabla \mathfrak{C}_{\tilde{c}}=0\,,
\eq
where we  use  the shorthand $\de \mathfrak{C}/\de \bfm\coloneqq\mathfrak{C}_\bfm$, $\de \mathfrak{C}/\de \rho \coloneqq\mathfrak{C}_\rho$, etc.  and summation of repeated indices is assumed. For the purpose at hand we assume $\mathfrak{C}$ is independent of $\bfm$, yielding the single condition
\bq
\rho\, \nabla \mathfrak{C}_\rho + \si\, \nabla \mathfrak{C}_\si + \tilde{c} \, \nabla \mathfrak{C}_{\tilde{c}}=0\,.
\label{cascon}
\eq
Equation \eqref{cascon} is satisfied by 
\bq
\mathfrak{C}=\int_\Om \calc(\rho, \si, \tilde{c})
\eq
for any  $\calc$ that is  Euler homogeneous of degree one, i.e., satisfies
\bq
 \calc(\la \rho,\la  \si, \la\tilde{c})=  \la\,  \calc(\rho, \si, \tilde{c})\,.
 \eq
The proof of this is straightforward.
 
 To complete the third step of our algorithm, the  entropy functional must be chosen from the set of Casimir invariants.  Writing the Euler homogeneous integrand as 
 \bq
 \calc(\rho, \si, \tilde{c})= \rho f(\si/\rho,\tilde{c}/\rho)
 \nn
 \eq
 it is clear that 
 \bq
 S=\int_\Om \rho \, s=\int_\Om\! \si \, 
 \label{entsi}
 \eq
 lies in our set of Casimirs.  This quantity was first shown to be a Casimir for the ideal fluid in  \citep{pjm82} and used for the thermodynamically consistent Navier-Stokes metriplectic  system in \citep{pjm84b}.  We note in passing, for other theories that might have a nontraditional dynamical equilibrium playing the role of thermodynamic equilibrium,  one may wish to choose another Casimir.

\subsection{Metriplectic 4-bracket for the Gibbs-Navier-Stokes system}
\label{ssec:SI4bkt}

Now let us turn to our fourth and final step of the algorithm, construction of the metriplectic 4-bracket.  To this end we review the formalism of \cite{pjmU23} in general terms and then apply it  to the GNS system, which is a   generalization of an example given in that work.

\subsubsection{General metriplectic 4-bracket dynamics}
\label{sssec:gen4bkt}

The metriplectic 4-bracket theory was  introduced by \citet{pjmU23} to describe the dissipative dynamic. Let us briefly recall the metriplectic 4-bracket description in infinite dimensions. In this description, we consider the dynamics of classical field theories with multi-component fields
\bq
\chi(z, t)=\left(\chi^1(z, t), \chi^2(z, t), \ldots, \chi^M(z, t)\right)
\label{chi}
\eq
defined on $z=(z^1,z^2,\dots, z^N)$ for times $t \in \mathbb{R}$. Here we use $z$ to be a label space coordinate with the volume element $ d^N\!z$, but with the domain unspecified. In fluid mechanics this domain would be $\Omega$,  the 3-dimensional domain occupied by the fluid and recall we used $\bfx$ for the  coordinate of this point. In general we suppose that $\chi^1,\ldots,\chi^M$ are real-valued functions of $z$ and $t$.  Given the space of functionals of $\chi$, $\calb$, we define 4-bracket as an
operator 
\bq
(\,\cdot\,,\,\cdot\,\,;\,\cdot\,,\,\cdot\,)\colon \calb \times \calb \times \calb \times \calb \rightarrow \calb
\eq
 such that for any four functionals $F,K,G,N\in\calb$ we have 
\bal
(F, G ; K, N)&=  \int  \! \!  d^N\!z \!\! \int \! \! d^N\!z'\,\! \! \int  \! \!  d^N\!z'' \!\! \int \! \! d^N\!z'''\,  \hat{R}^{\al\be\ga\de}
\nn \\
& \times \frac{\delta F}{\delta \chi^\al(z)} \frac{\delta G}{\delta \chi^\be\left(z'\right)} \frac{\delta K}{\delta \chi^\ga\left(z''\right)} \frac{\delta N}{\delta \chi^\de\left(z'''\right)}\,, 
\label{2b}
\eal
where  $\hat{R}^{\al\be\ga\de}(z,z',z'',z''')$ is a 4-tensor functional operator with coordinate form given by the following integral kernel:
\bal
 \hat{R}^{\al\be\ga\de}(z,z', z'',z''')[\chi]&=
 \\
 &\hspace{-1.25cm}\hat{R}(\mathbf{d}\chi^\al(z), \mathbf{d}\chi^\be(z'),
 \mathbf{d} \chi^\ga(z''), \mathbf{d} \chi^\de(z''))[\chi(z)]\, ,
 \nn
\eal
where $\al,\be, \ga,\de$ range over $1,2,\dots, M$.
The 4-bracket is assumed to  satisfy the following proprieties:\\
(i) Linearity in all arguments, e.g, for all $\la\in\R$
\bq
(F + \la H, K; G, N) = (F, K; G, N) + \la (H, K; G, N)
\label{2B1}
\eq
(ii) The algebraic  symmetries
\bal (F, K; G, N) &= -(K, F; G, N)
\label{2B2}\\
(F, K; G, N) &= -(F, K; N, G) 
\label{2B3}\\
(F, K; G, N) &= (G, N; F, K)
\label{2B4} 
\eal
(iii) Derivation in all arguments, e.g.,
\bq(F H, K; G, N) = F(H, K; G, N) + (F, K; G, N)H\,.
\label{2B6}
\eq
Here, as usual, $FH$ denotes point-wise multiplication. In addition, to ensure entropy production we require 
\bq
\dot{S}= (S,H;S,H) \geq 0\,.
\label{SHSH}
\eq
 Metriplectic 4-brackets that satisfy \eqref{2B1}--\eqref{SHSH} are called {\it minimal metriplectic}. 
In \S\S~\ref{sssec:KNcon} we will give a construction that ensures such appropriate positive semidefiniteness.

The  minimal metriplectic properties of metriplectic 4-brackets are reminiscent of the  algebraic properties possessed by a curvature tensor. In fact, every Riemannian manifold naturally has a metriplectic 4-bracket, and  $(S,H;S,H)$ provides a notion of sectional curvature  \citep[see][]{pjmU23}. 

%

From the metriplectic 4-bracket \eqref{2b}, the dissipative dynamics of an observable $o$ is generated as follows:
\bal
\p_t o&=\left(o, H ; S, H\right) =  \int  \! \!  d^N\!z \!\! \int \! \! d^N\!z'\,\! \! \int  \! \!  d^N\!z'' \!\! \int \! \! d^N\!z'''\,   \hat{R}^{\al\be\ga\de}
\nn\\
& \times \frac{\delta o}{\delta \chi^\al(z)} \frac{\delta H}{\delta \chi^\be (z')} \frac{\delta S}{\delta \chi^\ga (z'')} \frac{\delta H}{\delta \chi^\de(z''')}\,.
\label{d4B}
\eal
If we choose $o$ to be the Hamiltonian $H$, then $\dot{H}=\left(H, H ; S, H\right)\equiv 0$ by the antisymmetry condition of \eqref{2B2}.  If we choose $o$ to be the entropy $S$, then $\dot{S}=\left(S, H ; S, H\right)\geq 0$ by \eqref{SHSH}.

The dissipative dynamics generated by 4-bracket on our set of field variables $\chi$ is given by
\bal
\p_t{\chi}^\al(z)&=\left(\chi^\al, H ; S, H\right) 
\nn\\
&= \!\! \int \! \! d^N\! z'' \,   G^{\al\be} (z,z'')\frac{\delta S}{\delta \chi^\be(z'')}\,,
\label{chidyn}
\eal
where the $G$-metric is given as follows:
\bal
G^{\al\ga}(z,z'')&\coloneqq \int \! \! d^N\!z'\,\! \! \int  \! \!  d^N\!z'' \,  R^{\al\be\ga\de}(z,z',z'',z''')
\nn\\
&\hspace{1cm}\times \frac{\delta H}{\delta \chi^\be(z')} \frac{\delta H}{\delta \chi^\de(z''')}\,.
\eal
For the full metriplectic dynamics we would add the Poisson bracket contribution to the above. 
Equation \eqref{chidyn} is written so as to show that it amounts to a gradient system with the entropy  $S$ as generator.

\subsubsection{General Kulkarni-Nomizu construction}
\label{sssec:KNcon}

We can easily create specific metriplectic  4-brackets that have the minimal metriplectic properties:  the requisite symmetries and the positive semidefiniteness $(S,H;S,H)$. We do this by using the Kulkarni-Nomizu (K-N) product \citep{Kulkarni,nomizu1971spaces}. See also \citet{fiedler03} for relevant theorems.   Consistent with the bracket formulation of  \eqref{2b}, we deviate from the  conventional  K-N product  by working on the dual. Given two symmetric operator fields, say $\sum$ and $M$, operating on the variational derivatives; we again  use the subscript notation when convenient, 
\[
F_{\chi}\coloneqq\frac{\delta F}{\delta \chi} = \left(\frac{\delta F}{\delta \chi^1},\frac{\delta F}{\delta \chi^2},\ldots,\frac{\delta F}{\delta \chi^M}\right)\,,
\] 
the K-N product is defined as follows: 
\bal
  (\Sigma \wedge M)\left(dF,dK, dG,dN\right) 
&= \,\Sigma\left(dF, dG \right) M\left(dK, dN\right)
\nn\\
&\hspace{-.5cm} -\Sigma\left(dF , dN \right) M\left(dK , dG \right) \nn\\
&\hspace{-.5cm}+M\left(dF , dG \right) \Sigma\left(dK , dN \right)
\nn\\
&\hspace{-.5cm} -M\left(dF , dN \right) \Sigma\left(dK , dG \right)\,.
\eal
A finite-dimensional form of $\Si$ would be a  symmetric contravariant 2-tensor, say $\ga$,  and this would give the term
\bq
\ga(df,dg)=\ga^{ij}\frac{\p f}{\p z^i} \frac{\p g}{\p z^j}\,.
\label{sigmaf}
\eq
A conventional for K-N product would involve rank 2 covariant tensors. 
The form of \eqref{sigmaf}  suggests a general form  in infinite dimensions would be
\bq
\Si(dF,dG)= \!\!\int  \! \!  d^N\!z \!\! \int \! \! d^N\!z'\, \Si^{\al\be}(z,z') \frac{\delta F}{\delta \chi^\al(z)} \frac{\delta G}{\delta \chi^\be( z')}
\label{sigma}\,,
\eq
where $\Si^{\al\be}(z,z')$ is symmetric in both $\al,\be$ and $z, z'$ and operates to the right on both functional derivatives.  For example, 
\bq
\Si^{\al\be}(z,z')= L_{ab}^{\al\be}(z,z')\call^a \call'^b \,,
\eq
where $L_{ab}^{\al\be}$ is symmetric and $\call^a$ is a differential operator.  This implies, e.g., 
\bal
\Si(dF,dG)&= \int  \! \!  d^N\!z \!\! \int \! \! d^N\!z'\, L_{ab}^{\al\be}(z,z')
\nn\\
& \hspace{1cm}\times \call^a\frac{\delta F}{\delta \chi^\al(z)}\, \call'^b \frac{\delta G}{\delta \chi^\be(z')}\,.
\eal
With an expression  for $M$ similar to \eqref{sigma},  a term in the K-N decomposition would have the following form:
\bal
&\int  \! \!  d^N\!z \!\! \int \! \! d^N\!z'\,\! \! \int  \! \!  d^N\!z'' \!\! \int \! \! d^N\!z'''\, \Si^{\al\be}(z,z')\, M^{\ga\de}(z'',z''')
\nn\\
&\qquad\qquad  \times \  \frac{\delta F}{\delta \chi^\al(z)} \frac{\delta G}{\delta \chi^\be(z')}\,
 \frac{\delta K}{\delta \chi^\ga(z'')} \frac{\delta N}{\delta \chi^\de(z''')}\,,
\eal
which could be generalized further  by adding filtering kernels.

 It is easy to see that brackets constructed with this K-N product will have  all of the algebraic symmetries described in \S\S~\ref{sssec:gen4bkt}. In addition, it is shown in appendix \ref{A:seccurve} using the Cauchy-Schwarz inequality that positivity of $(S,H;S,H)$ is satisfied, if both $\Si$ and $M$ are positive semidefinite. Moreover, if one of $\Si$ or $M$  is positive definite,  defining an inner product, then the sectional curvature of (37) satisfies $(S,H;S,H)\geq 0$ with equality if and only if  $\de S/\de \chi \propto \de H/\de \chi$.  Thus, it is not difficult  to build minimal metriplectic 4-brackets.  
 
{Alternative to \eqref{sigma} we can define $\Si(dF,dG)$ pointwise  as
\bal
\Si(dF,dG)(z)&\coloneqq\! \int \! \! d^N\!z'\, \Si^{\al\be}(z,z') \frac{\delta F}{\delta \chi^\al(z)} \frac{\delta G}{\delta \chi^\be (z')}
\nn\\
&=  A^{\al\be}(z) \frac{\delta F}{\delta \chi^\al(z)} \frac{\delta G}{\delta \chi^\be (z)}
\label{sigma2}\,,
\eal
which could follow from \eqref{sigma} if we added an  additional argument to $\Si$.  Then,  with a corresponding form for $M$ the algebraic curvature symmetries would be induced in the integrand. This is the case for our  present purposes,  where we assume the}  specific K-N form given in \cite{pjmU23}, viz.\  where the  4-bracket is given by
\bq
(F, K ; G, N)=\!\!\! \int d^N\!z\, W\,(\Sigma \wedge M)\left(dF , dK , dG , dN \right)\,,
\nn
\eq
where  $W$ is an arbitrary weight, possibly depending on $\chi$ and $z$, that multiplies $(\Sigma \wedge M)$ where all of the functional derivatives are evaluated a the same point, $z$.  (See \eqref{M} and \eqref{Sgma} below.)
 In \S\S~\ref{sssec:GNS} we will see that this form  of 4-bracket is sufficient for the CHNS system of interest.

\subsubsection{Metriplectic 4-bracket for  the GNS system}
\label{sssec:GNS}

Now suppose our multi-component field variable $\chi$ is that for the  multiphase fluid, i.e., 
\bq
\psi(\bfx,t) = (\bfm(\bfx,t),  \rho(\bfx,t),\Tilde{c}(\bfx,t),{\sigma}(\bfx,t))
\label{psim0}
\eq
and  consider a  specific,  but still quite general, form of the K-N construction, one  adaptable to the GNS type of system.  For multi-component fields $\psi$ of our fluid we could choose
\bal
M(dF ,dG ) &=  F_{\psi^{\ga}}A^{\ga\de}G_{{\psi^{\de}}}\,,
\label{M}
\\
\Sigma (dF ,dG ) &=  
 \nabla F_{\psi^{\alpha}}\cdot B^{\alpha\beta}  \cdot \nabla G_{{\psi^{\beta}}} \,,
\label{Sgma}
\eal
where repeated indices are to be summed, $A^{\ga\de}=A^{\de\ga}$,  and in coordinates $B_{ij}^{\alpha\beta}$ is symmetric in both $i,j=1,2,3$ and  $\al,\be=1,\dots,6$, which are  indices that range over the six fields of $\psi$.  Here, the  nablas are contracted on $i$ and $j$.  
With the choices of \eqref{M} and \eqref{Sgma},  the metriplectic 4-bracket is 
\bal
(F,K;G,N) &=  \int_{\Omega}
\nabla F_{\psi^{\alpha}}\cdot B^{\alpha\beta} \cdot \nabla G_{{\psi^{\beta}}} \,  K_{\psi^{\ga}}A^{\ga\de} N_{{\psi^{\de}}} 
\nn\\
&- \nabla F_{\psi^{\alpha}}\cdot B^{\alpha\beta} \cdot \nabla N_{{\psi^{\beta}}}
\,  K_{\psi^{\ga}}A^{\ga\de}G_{{\psi^{\de}}}
\nn\\
& +  \nabla K_{\psi^{\alpha}}\cdot  B^{\alpha\beta} \cdot \nabla N_{{\psi^{\beta}}}
\, F_{\psi^{\ga}}A^{\ga\de}G_{{\psi^{\de}}} 
\nn\\
&- \nabla K_{\psi^{\alpha}}\cdot  B^{\alpha\beta} \cdot \nabla G_{{\psi^{\beta}}}
\,F_{\psi^{\ga}}A^{\ga\de}N_{{\psi^{\de}}}\,.
\label{ML4bkt}
\eal
Entropy production is governed by
\bal
(S,H;S,H) &=  \int_{\Omega}
\nabla S_{\psi^{\alpha}}\cdot B^{\alpha\beta} \cdot \nabla S_{{\psi^{\beta}}} \,  H_{\psi^{\ga}}A^{\ga\de}H_{{\psi^{\de}}} 
\\
&+  \nabla H_{\psi^{\alpha}}\cdot  B^{\alpha\beta} \cdot \nabla H_{{\psi_{\beta}}}
\, S_{\psi^{\ga}}A^{\ga\de}S_{{\psi^{\de}}} 
\nn\\
&- 2\,\nabla H_{\psi^{\alpha}}\cdot  B^{\alpha\beta} \cdot \nabla S_{{\psi_{\beta}}}
\,S_{\psi^{\ga}}A^{\ga\de}H_{{\psi^{\de}}}\,.
\nn
\eal
From the general results of appendix \ref{A:seccurve} it follows that $(S,H;S,H) \geq 0$ for this special case if $A^{\alpha\beta}$ and $B^{\alpha\beta}$ are positive semidefinte. 

Observe that \eqref{Sgma} could be replaced by the more general expression
\bq
\Sigma (dF ,dG ) =  
 \call F_{\psi^{\alpha}}\cdot B^{\alpha\beta}  \cdot \call G_{{\psi^{\beta}}} 
\label{Sgma2}\,,
\eq
where $\call$ is contained within a general class of pseudodifferential operators.  Later we  will  see an example of this.

Now consider an even more restrictive K-N product, a  special case of \eqref{chns4} with  what appears to be the simplest K-N options.  As discussed earlier, we do not expect our  4-bracket to depend on functional derivatives with respect to $\rho$, which could  produce density diffusion.  Thus, for  $M$ we take 
\bq
M(dF ,dG ) = F_{{\sigma}}G_{{\sigma}}.
\label{MSI}
\eq
The placement of the $\nabla$ in \eqref{chns4} leads to a diffusive type of relaxation, so this is natural, and the simplest case would be to select $\Si$ with no cross terms, i.e., 
\bal
{\Sigma}(dF ,dG ) &=  \nabla F_{{\bfm}} :\4tensLa_1 : \nabla G_{{\bfm}} + \nabla F_{{\sigma}} \cdot \tensLa_2\cdot  \nabla G_{{\sigma}} 
\nn\\
&\hspace{2cm}+ \nabla {F}_{\Tilde{c}}  \cdot \tensLa_3 \cdot  \nabla {G}_{\Tilde{c}} \,,
\label{SiSI}
\eal
where the 4-tensor $\4tensLa_1$ and the symmetric 2-tensors $\tensLa_2$ and $\tensLa_3$  are  to be determined. We make  the following choices
\bq
\4tensLa_1 = \frac{\4tensLa}{T},\qquad \tensLa_2 = \frac{\tenska}{T^2},\qquad \tensLa_3 = \frac{\tensD}{T}\,,
\label{TensorsLa}
\eq
where $\4tensLa$ is the isotropic Cartesian 4-tensor given by \eqref{viscous},  $\tensLa_{2,3}$ are symmetric positive definite 2-tensors defined by the previously introduced $\tenska$ and $\tensD$.  We take the weight $W$ to be the Lagrange multiplier defined in \S\S~\ref{sssec:genTh} i.e. $W =1$. Then, the 4-bracket reads
\bal
(F, K ; G, N)&= 
\label{4BSICHNS}
\\
&\hspace{-1.45cm} \int_{\Omega}\!\frac{1}{T}\Big[
\left[K_\sigma \nabla   F_{\mathbf{m}}-F_\sigma \nabla   K_{\mathbf{m}}\right]
:\4tensLa :
\left[N_\sigma \nabla   G_{\mathbf{m}}-G_\sigma \nabla   N_{\mathbf{m}}\right]
\nn\\
&\hspace{-1.0cm} +\ \frac{1}{T}\,
\big[K_\sigma \nabla  F_{{\sigma}}-F_\sigma \nabla  K_{{\sigma}}\big]
\cdot \tenska \cdot \big[N_\sigma \nabla  G_{{\sigma}}-G_\sigma \nabla  N_{{\sigma}}\big]
\nn\\
&\hspace{-1.0cm} + 
\big[K_\sigma \nabla F_{\Tilde{c}}  -F_\sigma  \nabla K_{\Tilde{c}}\big]\cdot
\tensD\cdot \big[N_\sigma \nabla G_{\Tilde{c}} - G_\sigma \nabla N_{\Tilde{c}}\big]
\Big]\,.
\nn
\eal

Upon insertion of  $H$ as given by \eqref{Hamc}  and $S$  from the set of Casimirs of \S\S~\ref{ssec:SIPbkt} to be as  in  \eqref{entsi}, the  dynamics is given by
\bq
\p_t\psi^\al = \{\psi^\al, H\} + (\psi^\al,H;S,H)\,. 
\eq
Using  ${H}_{\bfm} = \bfv$, ${H}_{\sigma} = T$, and ${S}_{\sigma} = 1$, the following GNS system is produced:
 \bal
    \p_t \mathbf{v}  &=\{\bfv,H\}+ (\bfv,H;S,H)
    \nn\\
    &= -\bfv\cdot\nabla \bfv -\nabla p/\rho
   + \frac1{\rho}\nabla\cdot(\4tensLa:\nabla \mathbf{v})\,,
   \label{Dpvt}
\\
 \p_t \rho  &=\{\rho,H\}+ (\rho,H;S,H)
 \nn\\
 &= -\bfv\cdot\nabla\rho   -\rho\, \nabla\cdot\mathbf{v} \,,
      \label{Dprt}
\\
    \p_t \Tilde{c} &=\{\tilde{c},H\}+ (\tilde{c},H;S,H)
    \nn\\
    &= -\bfv\cdot\nabla\tilde{c} -\Tilde{c}\,\nabla\cdot\mathbf{v}
 + \nabla\cdot(\tensD\cdot \nabla\mu)\,,
      \label{Dpct}
          \\
 \p_t \si &=\{\si ,H\}+ (\si,H;S,H)
 \nn\\
 &= -\bfv\cdot\nabla\si 
     - \sigma\,\nabla\cdot\mathbf{v}
       \nn\\ 
       &\qquad + \nabla\cdot\left(\frac{\tenska}{T}\cdot\nabla T\right) + \frac{1}{T^2}\nabla T\cdot\tenska \cdot \nabla T 
       \nn\\
       &\qquad + \frac{1}{T}\nabla\mathbf{v}:\4tensLa :  \nabla\mathbf{v} + \frac{1}{T} 
       \nabla\mu\cdot \tensD\cdot \nabla\mu \,.
      \label{Dpst}
\eal

By construction we automatically have energy conservation, i.e., for  \eqref{Hamc}  $\dot H=0$,   and entropy production 
\bal
\dot{{S}}&= ({S}, {H} ; {S}, {H})
\nn\\
&= \int_\Om \frac{1}{T}\bigg[ \nabla\mathbf{v}: \4tensLa :\nabla\mathbf{v} + \frac{1}{T}\nabla T \cdot \tenska\cdot \nabla T
\nn\\
& \hspace{3.10cm}+   \nabla\mu\cdot \tensD\cdot \nabla\mu \bigg] \geq 0\,.
\label{dots2}
\eal

\subsection{GNS metriplectic 2-bracket  and conventional  fluxes and affinities}
\label{ssec:SI2bkt}

For completeness we demonstrate two things in this subsection:  how the metriplectic 4-bracket formalism relates to the original binary metriplectic formalism  given in \citet{pjm84,pjm84b,pjm86} and how it relates to conventional nonequilibrium thermodynamics, making the connection between the 4-bracket K-N construction and the phenomenology of thermodynamics fluxes and affinities (sometimes called  thermodynamic forces).

\subsubsection{GNS metriplectic 2-bracket}
\label{sssec:genTh}

As noted above  we are concerned with the metriplectic dynamics introduced in \citet{pjm84,pjm84b,pjm86} \citep[see also][]{pjm09a,pjmC20}, but we mention that other binary brackets for describing dissipation were presented over the years  \citep[e.g.][]{pjmK82,kaufman1984dissipative, pjmH84,grm84,og97,Beris94,Edwards98}. In addition we mention a recent  alternative  approach to multiphase fluids, one based on constrained variational principles,  that is given in \citet{Eldred20}.  We refer the reader to \citet{pjmU23}  for comparisons with other formulations and how they emerge from the metriplectic 4-bracket.  

Metriplectic dynamics was introduced  as a means of building thermodynamically consistent theories in terms of a binary bracket, which we now call the metriplectic 2-bracket.  The theory applies to a wide class of dynamical systems, including both ordinary and  partial differential equations.  Evolution of an observable $o$ using  the metriplectic 2-bracket  has the following form:
\begin{eqnarray}
  \p_t o   = \{o,\calf\} - (o,\calf)_H\,,
    \label{2bkto}
\end{eqnarray}
where as before  $\{\,.\,\}$  is the noncanonical Poisson bracket that generates the ideal part of the dynamics, while now $(F,G)_H$, the metriplectic 2-bracket,  generates the dissipative part.  The functional $\calf$ represents  the global Helmholtz free energy of the system, and is given by:
\begin{eqnarray}
\calf =  {H} - \mathcal{T} {S}  \,,   
\end{eqnarray}
where again ${H}$ is the Hamiltonian and  ${S}$ the entropy selected from the set of Casimirs of the noncanonical Poisson bracket (ensuring  $\{F,{S}\} = 0 \,$ for any functional $F$), and $\mathcal{T}$ is a  uniform nonnegative constant (a global temperature). The  metriplectic 2-bracket  $(\,,\,)$ is assumed to be bilinear, symmetric, and  satisfies
\begin{eqnarray}
    (F,{H})_H \equiv 0 \quad\text{ for any functional }F.
\end{eqnarray}
Thus,  metriplectic systems are thermodynamically consistent:
\begin{description}
    \item[First law] \ 
     (energy conservation):
    \bal
    \dot{H} &= \{ {H}, \calf\} -  ( {H}, \calf)_H
    \nn\\
    &= \{ {H}, {H}\} + \mathcal{T}( {H}, {S})_H = 0\,;
      \label{engcon}
      \eal
    \item[Second law]\ (entropy production):
\bal
  \dot{S} &=
  \{ S, \calf\} -  ( S, \calf)_H
  \nn\\
  &=  -( {S}, {H})_H + \mathcal{T}( {S}, {S})_H = \mathcal{T}( {S}, {S})_H \geq 0\,,
\label{entprod}
\eal
\end{description}
which follows because $ \{ S, \calf\}\equiv 0$ and $( {S}, {H})_H\equiv 0$.  As shown in \citet{pjmU23} the metriplectic 2-bracket emerges from the 4-bracket as follows:
\bq
(F, {G})_{ {H}}=(F,{H}; {G}, {H})\,, 
\eq
where for convenience here and henceforth  we set $\mathcal{T}=1$. Because of the minimal metriplectic properties of the 4-bracket, we are assured to have the thermodynamic consistency of \eqref{engcon} and \eqref{entprod}. 

The 2-bracket that emerges from the general 4-bracket of \eqref{ML4bkt} is the following:
\bal
(F,G)_H &= (F,H;G,H)
\nn\\
 &=  \int_{\Omega}
\nabla F_{\psi^{\alpha}}\cdot B^{\alpha\beta} \cdot \nabla G_{{\psi^{\beta}}} \,  H_{\psi^{\ga}}A^{\ga\de}H_{{\psi^{\de}}} 
\nn\\
&- \nabla F_{\psi^{\alpha}}\cdot B^{\alpha\beta} \cdot \nabla H_{{\psi^{\beta}}}
\,  H_{\psi^{\ga}}A^{\ga\de}G_{{\psi^{\de}}}
\nn\\
&+  \nabla H_{\psi^{\alpha}}\cdot  B^{\alpha\beta} \cdot \nabla H_{{\psi_{\beta}}}
\, F_{\psi^{\ga}}A^{\ga\de}G_{{\psi^{\de}}} 
\nn\\
&- \nabla H_{\psi^{\alpha}}\cdot  B^{\alpha\beta} \cdot \nabla G_{{\psi_{\beta}}}
\,F_{\psi^{\ga}}A^{\ga\de}H_{{\psi^{\de}}}\,,
\label{chns4}
\eal
which  in light of the K-N product satisfies $(F,H)_H=0$ for all $F$.  This will be true for any choice of the Hamiltonian $H$.  Indeed, a special case of this was used in \eqref{4BSICHNS} to obtain the thermodynamically consistent set of equations \eqref{Dpvt}-\eqref{Dpct}. Recall, for this case $M$ and $\Si$ were chosen as in \eqref{MSI} and \eqref{SiSI}, special cases of \eqref{M} and \eqref{Sgma}.

Another 2-bracket can be obtained from  \eqref{ML4bkt} by making  a convenient choice of variables;  viz.,   instead of  the variables $\psi$ of \eqref{psim} in \eqref{chns4} we choose the  following:
\bq
\xi(\bfx,t) \coloneqq (\bfm(\bfx,t),\rho(\bfx,t),\Tilde{c}(\bfx,t),e(\bfx,t))\,,
\label{psim1}
\eq
where the  total energy density is used  instead of the entropy density as one  of our dynamical variables.  That this is possible is well known in thermodynamics  because the entropy must be a monotonic increasing function of the internal energy, which allows via the inverse function theorem   transformation between the extensive energy  or  extensive entropy representations \citep{callen}. We will denote these density variables in order by $\xi_\al$, $\al=1, \dots,6$. With this choice the Hamiltonian   \eqref{Hamc} is given by
 \bq
H = \int_{\Omega} e=\int_{\Omega} \xi_6\,,
\eq
 $\de H/\de \xi_\al= \de_{\al 6}$, the Kronecker delta, and $\nabla\de H/\de \xi_\al\equiv 0$.  Thus, \eqref{chns4} reduces to 
\bq
(F,G)_H =  (F,H;G,H)
 =  \int_{\Omega}
\nabla F_{\xi_{\alpha}}\cdot L_{\alpha\beta} \cdot \nabla G_{{\xi_{\beta}}} \,   ,
\label{2bktxi}
\eq
where without loss of generality we set $A^{66}=1$ and    $B^{\al \be}= L_{\alpha\beta}$ for this case.   In \S\S~\ref{ssec:fandf} we will follow  \citet{pjmC20} and show how the bracket of \eqref{2bktxi} fits into the framework of conventional  nonequilibrium thermodynamics as, e.g.,  described in \citet{de2013non}. In this way we will physically identify the meaning of  $L_{\al\be}$.

\subsubsection{Fluxes and affinities for  the GNS system}
\label{ssec:fandf}

A fundamental equations of nonequilibrium thermodynamics is the general thermodynamic identity 
\bq
    \text{d}\sigma = X^{\alpha} \text{d} \xi_{\alpha},
    \label{sgmxi}
\eq
which relates  $\sigma$,  the entropy density, to  the  $\xi_{\alpha}$ densities associated with conserved extensive properties and to  $X^{\alpha} \coloneqq \partial \sigma /\partial \xi_{\alpha}$, quantities called affinities (or thermodynamic forces).  All the densities are characterized by the following conservation equations:
\bq
    \partial_t \xi_{\alpha} + \nabla \cdot\mathcal{J}_{\alpha} = 0\,,
\eq
where $\mathcal{J}_{\alpha}$ is at  present an unknown flux associated with the density $\xi_{\alpha}$. Then, the evolution of the entropy is given by
\bq
    \partial_t \sigma + \nabla\cdot \left(
    X^{\alpha} \mathcal{J}_{\alpha} \right) =   \mathcal{J}_{\alpha}\cdot\nabla X^{\alpha} \,.
    \label{split} 
\eq
 The righthand side of \eqref{split} is the  dissipative term, which is the sum of  the  fluxes $\mathcal{J}_{\alpha}$  contracted with  $\nabla X^{\alpha}$. The linear assumption of nonequilibrium processes amounts to relating fluxes and affinities according to 
\bq
    \mathcal{J}_{\alpha} = L_{\alpha \beta}   X^{\beta}\,.
  \label{fluxfo}
\eq
If we identify the $L_{\al\be}$ of \eqref{fluxfo} with that of \eqref{2bktxi}, we see how metriplectic brackets are related to the flux-affinity relations. Onsager symmetry, assumed to arise from microscopic reversibility, amounts to the symmetry $ L_{\alpha \beta} = L_{\be \alpha}$ and the semi-definiteness property assures the second law, i.e., entropy growth. It remains to identify how the fluxes of \S\S~\ref{ssec:SIDescription} enter the picture.

To further identify the meaning of $L_{\al\be}$ we revisit the thermodynamics of \eqref{TH1},  in light of our choice of the variables  $\xi$ of \eqref{psim}.  Thus, we rewrite the  thermodynamic relation \ref{TH1} upon changing variables, 
\begin{equation}
    T \text{d}\sigma = \text{d} e - \mathbf{v}\cdot\text{d}\mathbf{m} - \mu\text{d}\Tilde{c} 
    - g\text{d}    \rho\,,
    \label{tdsgibbs}
 \end{equation} 
where $e$ is  the energy density of  \eqref{e} and  
$g$ is a modified specific Gibbs free energy, viz. 
\begin{equation}
g:= u - Ts +  {p}/{\rho} - \mu c -  {|\mathbf{v}|^2}/{2}\,.
\end{equation}
We have assumed in \eqref{glbM} that there is no flux associated with $\rho$, i.e., in CHNS  chemical reactions and/or particle creation and annihilation are ignored.  Thus, the
phase space for the thermodynamics is smaller than that for the Hamiltonian dynamics,  because the variable $\rho$  as seen e.g. in \eqref{Dprt} has no dissipative terms.  This leads us to focus on the thermodynamic variables $(e, \mathbf{m},\Tilde{c})$ and \eqref{tdsgibbs} reduces to 
\begin{equation}
    T \mathrm{d}\sigma = \mathrm{d}e - \mathbf{v}\cdot\mathrm{d}\mathbf{m} -\mu \mathrm{d} \Tilde{c}\,.
    \label{TH2}
\end{equation}

Comparison of \eqref{sgmxi} and \eqref{TH2} suggests we require the  affinities associated with  $\mathbf{m}$, $e$, and $\Tilde{c}$.
The conventional choices for these affinities are $\nabla(1/T)$, $\nabla(-\bfv/T)$, and $\nabla(-\mu/T)$, respectively \citep{de2013non}.  However, examination of  \eqref{entropy1} or \eqref{dots2} suggests using instead $\nabla T$, $\nabla\bfv$, and $\nabla\mu$, as was done in \citet{pjmC20}.

The relationship between the  flux-affinity relations in terms of these two choices of bases are given by the following:
\bal \bar{J}_{\mathbf{m}}&=L_{\bfm e}\cdot \nabla \left({\frac{1}{T}}\right)+L_{\bfm\bfm} :\nabla \left(\frac{-\mathbf{v}}{T}\right) 
+L_{\bfm\Tilde{c}} \cdot \nabla \left({\frac{-\mu}{T}} \right)
\nn\\
& \hspace{.5 cm} = - \4tensLa : \nabla \mathbf{v}\,, 
\label{L1}
\\
{\bfJ}_{e}&=L_{ee}\cdot \nabla \left(\frac{1}{T}\right) +L_{e \bfm}:  \nabla\left(\frac{-\mathbf{v}}{T}\right) 
+L_{e\Tilde{c}}\cdot \nabla\left(\frac{-\mu}{T}\right)
\nn\\
& \hspace{.5 cm}= -\mathbf{v} \cdot (\4tensLa : \nabla \mathbf{v}) - \tenska\cdot \nabla T - \mu \tensD\cdot \nabla \mu\,, 
\label{L2}\\
 {\bfJ}_{\Tilde{c}}&=L_{\Tilde{c} {e}} \cdot \nabla\left(\frac{1}{T}\right) +L_{\Tilde{c} \bfm} : \nabla\left(\frac{-\mathbf{v}}{T}\right) 
 +L_{\Tilde{c}\Tilde{c}}\cdot  \nabla \left(\frac{-\mu}{T}\right)
 \nn\\
 & \hspace{.5 cm} = -\tensD\cdot \nabla \mu 
 \label{L3}\,.
 \eal
 Recall $\bar{J}_\bfm$ is a 2-tensor, thus $L_{\bfm e} =L_{e\bfm} $ is a 3-tensor, $L_{\bfm\bfm}$ is a 4-tensor, and $L_{\bfm\Tilde{c}}=L_{\Tilde{c}\bfm}$ is a 3-tensor.  Since $\bfJ_e$ and $\bfJ_{\tilde c}$ are vectors,   $L_{ee},  L_{{e}\Tilde{c}}=L_{\Tilde{c}e}$,  and $L_{\Tilde{c}\Tilde{c}}$ are 2-tensors. 
 From \eqref{L1},  \eqref{L2}, and  \eqref{L3},  we identify the components of $L_{\al\be}$ as follows:
\bal
    &L_{\mathbf{m}e} = T \4tensLa \cdot \mathbf{v}\,,\ \ \ 
    L_{\mathbf{m}\mathbf{m}}  = T \4tensLa\,,\ \ \ 
    L_{\mathbf{m}\Tilde{c}}  = 0\,,\ \ \  
    L_{{e}\Tilde{c}}  = T\mu \tensD'\,, 
  \nn\\ 
    &L_{ee} = T^2\tenska  + T \mathbf{v}\cdot\4tensLa\cdot \mathbf{v} + T\mu^2 \tensD\,, \quad 
L_{\Tilde{c}\Tilde{c}}  = T \tensD\,.
\label{tensorL2}
\eal
The metriplectic 2-bracket in terms of the $\xi$-variables is given by 
\bal
     (F,G)_H &=  \int_{\Omega}T\, \Big[ \nabla F_{e} \cdot (T\tenska  + \mathbf{v}\cdot\4tensLa\cdot \mathbf{v} + \mu^2 \tensD)\cdot \nabla G_{e} 
     \nn \\
     &+  \nabla F_\bfm: \4tensLa :\nabla G_\bfm + \nabla F_{\Tilde{c}} \cdot \tensD\cdot \nabla G_{\Tilde{c}} 
     \nn\\
    & + \nabla F_e\cdot (\4tensLa \cdot \mathbf{v}) : \nabla  G_\bfm +\nabla G_e \cdot (\4tensLa\cdot \mathbf{v}): \nabla  F_\bfm
   \nn \\
    & + \mu\,\nabla F_e \cdot  \tensD \cdot \nabla  G_{ \Tilde{c}} 
    + \mu\, \nabla G_e \cdot  \tensD \cdot\nabla F_{\Tilde{c}}\Big]\,.
    \label{Dissbracket}
\eal
 Upon writing
 \bq
 S=\int_\Om \si(\rho, e, \Tilde{c},\bfm)
 \eq
 and using standard thermodynamic manipulations we obtain
    \begin{eqnarray}
     S_e =  {1}/{T}\,,\ \  S_\bfm = -{\mathbf{v}}/{T}\,,\ \ \mathrm{and}\ \    {S}_{\Tilde{c}} = - {\mu}/{T}\,. 
    \end{eqnarray}
  Inserting these  into $(o,S)_H$ using the 2-bracket of  \ref{Dissbracket} yields the dissipative terms of \eqref{Dpvt}, \eqref{Dpct}, and with the manipulations of transforming from $e$ to $\si$, those of the entropy equation  \eqref{Dpst}.  By direct calculation, as well as by construction, we obtain $(H,S)_H=0$  and $\dot{S}=  ( {S}, {S})_H = (S,H;S,H)\geq 0$ which reproduces \eqref{dots2}.

To close the circle we transform the bracket of \ref{Dissbracket} in terms of the variables $(e,\mathbf{m},\Tilde{c})$ to one in terms of  $(\sigma,\mathbf{m},\Tilde{c})$ via the following chain rule formulas:
\bal
G_e &\rightarrow  G_{\sigma}/T,\qquad G_\bfm \rightarrow G_\bfm -\mathbf{v}\,,  G_\si/ {T}\,,
\nn\\
& \quad G_{\Tilde{c}} \rightarrow G_{\Tilde{c}}  -{\mu} \,   G_\si /T   \,.
\eal
This calculation gives precisely the bracket of \eqref{4BSICHNS}.

\section{The Cahn-Hilliard-Navier-Stokes system}
\label{sec:DICHNS}

 {Now we apply our algorithm to obtain the Cahn-Hilliard-Navier-Stokes system (CHNS) which allows for diffuse-interfaces. 
 We follow the steps in order,  just as in \S~\ref{sec:GNS}. However, here we have the additional step of aligning the desired entropy functional with the Poisson bracket, so that it is indeed a Casimir invariant.}

\subsection{Hamiltonian and Entropy functional forms}
\label{ssec:chnsHS}

The phenomenon of material transport along an interface is known as the Marangoni effect. The presence of a surface tension gradient naturally induces the migration of particles, moving from regions of low tension to those of high tension. This gradient can be triggered by a concentration gradient (or also a temperature gradient). In two-phase theory the interface between phases is regarded as being diffuse. According to the work of \citet{taylor-cahn98}, one can model the diffuse interface by a single order parameter, say $\phi$ and with a free energy functional,  
\bq
\mathfrak{F}= \int_{\Omega} \frac{\epsilon}{2}\, \Gamma^2(\nabla \phi) + \frac{1}{\epsilon} V(\phi)\,,
\eq
with $\Gamma$ being a homogeneous function of degree one, further details on this will be provided later. Here $V$ can be any non-negative function that equals zero at $\phi = \pm1$ and  $\epsilon$ is a small parameter that goes to zero in the sharp-interface limit. We choose the order parameter $\phi$ to be the concentration.

In the isotropic surface energy case \citet{Guo} develop a phase-field model for two-phase flow, which is thermodynamically consistent. The modeling based on a non-classical choice of energy  and entropy, given respectively by  
\bal
H_{GL} &= \int_{\Omega} \frac{\rho}{2}\,   |\bfv|^{2} + \rho u(\rho,s,c) + \frac{\rho}{2} \lambda_u|\nabla c|^2\,,
\\
S_{GL} &= \int_{\Omega} \rho s + \frac{\rho}{2} \lambda_s|\nabla c|^2\,, 
\eal
where $u$ and $s$ stand for the classical  specific internal energy and entropy, respectively,   while  the coefficients $\lambda_s$ and $\lambda_u$ are constant parameters. 

Alternatively,  \citet{anderson2000phase} propose a model of phase-field of solidification with convection, the model permits the interface to have an anisotropic surface energy. The choice of energy and entropy are given by
\bal
H_{AMW} &= \int_{\Omega} \frac{\rho}{2}\, |\bfv|^2 + \rho u(\rho,s,c) 
+  \frac{\epsilon_E^2}{2}\Ga^2(\nabla c)\,,
\\
S_{AMW}&=\int_\Om \rho s - \frac{\ep_S^2}{2} \Ga^2(\nabla c) \,,
\eal
where the coefficients $\epsilon_S$ and $\epsilon_E$ are assumed to be constant and $\Ga$ is a homogeneous function of degree one that takes a vector to a real number.

In this section, we  explore a choice of energy and entropy functionals, from which the previously mentioned choices are special cases, and we consider the associated free energy functional, viz.
\bal
H^a &= \int_{\Omega} \frac{\rho}{2} |\mathbf{v}|^2 + \rho u  + \frac{\rho^{a} }{2} \lambda_{u} \Gamma^2(\nabla c) 
 \eqqcolon \!\int_{\Omega} e_{\mathrm{Total}}^a\,,
\label{TE1}
\\
\label{CS2}
{S^a} &= \int_{\Omega} \rho s + \frac{\rho^{a} }{2} \lambda_{s} \Gamma^2(\nabla c) 
\eqqcolon \!\int_{\Omega} \si_{\mathrm{Total}}^a\,,
\\
\mathfrak{F}^a\ &= \int_{\Omega} \rho f  
+ \frac{ \rho^{a}}{2} \lambda_{f}(T) \Gamma^2(\nabla c) \,,
\label{calf}
\eal
where $u$, $s$ and $f$ stand for the classical specific internal energy, entropy,  and free energy,  respectively, the coefficients $\lambda_s$ and $\lambda_u$ are constant parameters, and $\lambda_f(T)$ is a parameter depending on the temperature that  will lead to anisotropic surface energy effects.  He have defined the total densities $e^a_{\mathrm{Total}}$ and $\si^a_{\mathrm{Total}}$ for later use.  The parameter $a$ takes on  two values:  $a = 0$  reduces  \eqref{TE1} and  \eqref{CS2}  to  the expressions of \citet{anderson2000phase}, {where we set $\ep^2_{E}=\la_{u}$ and  $\ep^2_{S}=-\la_{s}$,} while for  $a= 1$  they reduce to those used by  \citet{Guo} provided  the choice of an isotropic surface energy is assumed, viz.,   $\Gamma (\nabla c) = |\nabla c|$.    {Thus, as is clear from  \eqref{CS2}, \eqref{TE1}, and \eqref{calf} that  the dimensions of $\la_f$, $\la_s$, and  $\la_u$ are either specific or volumetric depending on the case.  As usual,  we have the thermodynamic relation
\bq
f= u - Ts\,,
\eq
which allows us to assume the relationship between the coefficients 
\bq
    \lambda_f (T) = \lambda_u - T\lambda_s \quad\text{   and   }\quad \frac{\text{d}\lambda_f(T)}{\text{d}T} 
    = -\lambda_s\,.
\eq

To summarize, our expressions \eqref{TE1}  and  \eqref{CS2}     generalize the model studied by  \citet{Guo} by including $\Ga$, which accounts for  anisotropic surface energy effects, while  our expressions  generalize   the model of \citet{anderson2000phase} by including the factors of $\rho$ in the integrands  making  all quantities in the integrands specific quantities multiplied by the density, giving rise to more general sources of energy.}

Because $\Gamma$ is a homogeneous function of degree unity, 
\bq
\Gamma (\lambda \bfp ) = \lambda \Gamma(\bfp ) \text{  for all   } \lambda > 0\,.
\label{Ehomo1}
\eq
Differentiating  \eqref{Ehomo1} with respect to $\la$ and then setting $\la=1$  yields the fundamental relation
\bq
 \Gamma(\bfp ) = \bfp\cdot \bfxi \coloneqq p_j  \frac{\partial \Gamma(\bfp)}{\partial p_j}\,.
 \label{GF1}
 \eq
Then, differentiating \eqref{GF1} gives a second well-known relation, 
\bq
\frac{\p \Ga}{\p p_i}=\frac{\p}{\p p_i} ( \bfxi\cdot \bfp)= \xi_i + \frac{\p^2 \Ga}{\p p_i \p p_j}p_j
= \xi_i\,,
\label{GF2}
\eq
where evidently  $p_j$ must be  a null eigenvector of the matrix $ {\p^2 \Ga}/{\p p_i \p p_j}$. 
Henceforth we will assume the argument of $\Ga$ to be $\nabla c$.  For the case of isotropic surface energy,  where $\Ga(\nabla c)= |\nabla c|$, the  associated homogeneous function of degree zero is given by 
\bq
{\bfxi} = \nabla c/|\nabla c|\,.
\eq

From \eqref{calf} we can obtain a generalized chemical potential
\bal
\mu^a_\Ga&\coloneqq \frac{\de \mathfrak{F}^a}{\de \tilde{c}}=\rho\frac{\p u}{\p \tilde{c}} -\nabla\cdot(\la_f \rho^a\Ga \nabla c)
\nn\\
&=\mu -\frac1{\rho} \nabla\cdot(\la_f \rho^a\Ga \bfxi)\,,
\eal
where recall $\tilde{c}=\rho c$.  For the case of isotropic surface energy, this becomes
\bq
\mu^a_{|\nabla c|}=\mu -\frac1{\rho} \nabla\cdot(\la_f \rho^a\nabla c)\,.
\eq
Upon setting  $a=1$   \citep[the case of][]{Guo},  this reduces to 
\bq
\mu^1_{|\nabla c|}=\mu -\frac1{\rho} \nabla\cdot(\la_f \rho\,  \nabla c)\,, 
\eq
an expression that differs from that in \citet{Guo} unless $\la_f \rho$ is constant. If this is the case and we choose a  classical $\mu=c^3-c$, corresponding to  the quartic  Laudau potential,   we obtain
\bq
\mu_{CH}= c^3-c - {\la_f } \nabla^2 c\,, 
\eq
 the  chemical potential of  Cahn and Hilliard who indeed make these assumptions  \citep[cf.\ page 267 of][]{cahn1958free}.   For  $\al=0$  \citep[the case of][]{anderson2000phase}, we have 
\bq
\mu^0_\Ga =\mu -\frac1{\rho} \nabla\cdot(\la_f \,\Ga \bfxi)\,,
\eq
 which allows for the weighted mean curvature effects of anisotropy.
 
Maintaining the same set of dynamical variables as in \eqref{psim0} (or an equivalent set) and making  the choices of energy and entropy functionals of \eqref{TE1} and \eqref{CS2},  we have completed the first two steps of the algorithm.

\subsection{Noncanonical Poisson bracket of the  Cahn-Hilliard-Euler system}
\label{ssec:chnsPB}

To complete the next step of the algorithm, the third step, we need  to manufacture a bracket that has \eqref{CS2} as a Casimir invariant. We do this by starting from the bracket of \eqref{PB1} in terms of the original variables $\Psi=\{\bfm,\rho, \tilde{c},\si\}$ and then transforming it to a new set of dynamical variables  $\hat{\Psi}^a\!\coloneqq \{\hat\bfm,\hat\rho, \hat{\tilde{c}},\hat{\si}^a \}$, giving the same bracket in terms of new coordinates.  We have included the superscript $a$ because in effect we have two sets of coordinates, corresponding to the desired entropies of \eqref{CS2} for $a=0$ and $a=1$.  To distinguish the old from the new,  we write the bracket in the transformed variables as  $\{\hat{F},\hat{G}\}^a$.    Because of coordinate invariance, $\{C,F\}= \{C^a,\hat{F}\}^a=0$, where $F[\Psi]=\hat{F}[\hat\Psi^a]$ is any functional written in one or the other coordinates. The Casimir $S=\int_\Om \si$ in our original coordinates is transformed into a different form in  the new coordinates. Specifically, we change the  variables  as follows:
\bal
{\mathbf{m}} &= \hat{\mathbf{m}}\,,\quad 
  {\rho} = \hat{\rho}\,,\quad {\Tilde{c}}  = \hat{\Tilde{c}}\,,
  \nn\\
   \sigma &= \hat{\sigma}^a  +  \frac{\hat{\rho}^a}{2}  \lambda_{s} \Gamma^2(\nabla \hat{c}) \,, 
  \label{hattrans}
\eal
where  $\hat{\tilde{c}}=\hat\rho \hat{c}$. 
Consequently, the entropy $S$ in the old coordinates written in terms of the new coordinates  will, by design,  become the following Casimir for the Poisson bracket in the new coordinates:
\bq
\hat{S}^a=\int_\Om \hat{\sigma}^a  +  \frac{\hat{\rho}^a}{2}  \lambda_{s} \Gamma^2 (\nabla\hat{c})\,.
\eq
Thus we have manufactured a bracket with the entropy expression of \eqref{CS2} as a Casimir.

Transformation of the Poisson bracket \eqref{PB1} requires use of the functional chain rule. 
For convenience we use
\bq
   \hat{\sigma}^a  = \sigma -  \frac{\rho^a}{2}  \lambda_{s} \Gamma^2(\nabla c)
\eq
and consider  the variation of any functional of  the new variables.  Thus we use   $\delta \rho = \delta \hat{\rho}$, $\delta \mathbf{m} = \delta \hat{\mathbf{m}}$, $\delta \Tilde{c} = \delta \hat{\Tilde{c}}$, and for the entropy variable 
\bal
\delta \hat{\sigma}^a &= \delta \sigma - \frac{1}{2}a \rho^{a-1}\lambda_{s}\Gamma^2(\nabla c)\delta \rho - \rho^{a} \lambda_{s} \Gamma{\bfxi}\cdot\nabla \delta\left(\frac{\Tilde{c}}{\rho}\right)
\nn\\
&= \delta \sigma - \frac{1}{2}a \rho^{a-1}\lambda_{s}\Gamma^2(\nabla c)\delta \rho  - \rho^{a} \lambda_{s} \Gamma {\bfxi}\cdot \nabla \left(\frac{\delta \Tilde{c}}{\rho}  \right)
\nn\\
& \hspace{2 cm} + \rho^{a} \lambda_{s} \Gamma {\bfxi}\cdot\nabla\left(\frac{\Tilde{c}}{\rho^2} \delta \rho\right)\,, 
\eal
where use has been made of \eqref{GF2}.   Now let $F$ be an arbitrary functional of the original variables and $\hat{F}$ the same  functional in terms of the new variables.  Thus,  
\bal
\int_{\Omega} {\hat{F}}_{\hat{\bfm}} \cdot\, \delta\hat{\mathbf{m}}
&+ \hat{F}_{\hat\rho} \, {\delta \hat{\rho}}  +  \hat{F}_{\hat{\sigma}^a} \,\de {\hat{\sigma}^a} +  \hat{F}_{\hat{\Tilde{c}}}\,
\delta\hat{\Tilde{c}} 
\\
&= \int_{\Omega} F_\bfm \cdot \delta\mathbf{m} +F_\rho\, \delta{\rho} + F_\si\, \delta{\sigma} +  F_{\Tilde{c}}\, \delta{\Tilde{c}}\,.
\nn
\eal
Note, no sum over $a$ is to be done.  By identification of terms we obtain
\bal
    F_\bfm &=  \hat{F}_{\hat{\bfm}}\,,  \quad  F_\si  =  \hat{F}_{\hat{\sigma}_a}\,,
     \nn \\
    F_{\rho} &= \hat{F}_{\hat{\rho}} 
    -  \frac{a}{2} \hat{\rho}^{a-1}\lambda_{s}\Gamma^2  \,  \hat{F}_{\hat{\sigma}^a}
    -  \frac{\hat{\Tilde{c}}}{\hat{\rho}^2}\,  \nabla\cdot\left(\hat{\rho}^{a}\lambda_s \Gamma{\bfxi}\, 
     \hat{F}_{\hat{\sigma}^a} \right)\,,
    \nn\\
    F_{\Tilde{c}} &=  \hat{F}_{\hat{\Tilde{c}}} + \frac1{\hat{\rho}} \nabla\cdot\left(\hat{\rho}^{a}\lambda_s  \Gamma{\bfxi}  
    \hat{F}_{\hat{\sigma}^a} \right)\, .
    \label{var2}
\eal 
 The transformed Poisson bracket  is obtained by inserting the expressions of \eqref{var2} into \eqref{PB1}, writing it entirely  in terms of the hat variables. Upon doing this and then dropping the hats,   we get for any functionals ${F}$ and ${G}$  the following Poisson bracket:
\bal
\{{F},G\}^a &= -  \int_{\Omega} {\mathbf{m}} \cdot\Big[F_\bfm \cdot\nabla G_\bfm -  G_\bfm\cdot\nabla  F_\bfm \Big] 
\nn\\
&+ {\rho}\,  \Big[ 
F_\bfm \cdot \nabla \big( G_\rho 
 - {a} \rho^{a-1}\lambda_{s}  \Gamma^2  G_{\si^a}/2
\nn\\
& \hspace{2cm} - {\Tilde{c}} \,  \nabla\cdot
   (\rho^{a}\lambda_s \Gamma{\bfxi} G_{\si^a} )/\rho^2 \big)
\nn\\
&\hspace{.5cm}- G_\bfm \cdot\nabla  \big(F_\rho  - {a} \rho^{a-1}\lambda_{s} \Gamma^2  F_{\si^a} /2 
\label{GenBkt}\\
& \hspace{2cm}  -  {\Tilde{c}} \, \nabla\cdot\left(\rho^{a}\lambda_s \Gamma{\bfxi} F_{\si^a} \right)/\rho^2\big) \Big] 
\nn\\
&+ \left(
{\sigma}^a  +  {{\rho}^a}  \lambda_{s} \Gamma^2/2
\right)
\Big[F_\bfm \cdot\nabla G_{\si^a}  - G_\bfm \cdot\nabla F_{\si^a} \Big] 
\nn\\
&+ \Tilde{c}\,  \Big[ F_\bfm \cdot\nabla \left(G_{\Tilde{c}}
+  \nabla\cdot\left(\rho^{a}\lambda_s  \Gamma{\bfxi} G_{\si^a} \right)/\rho \right)  
\nn\\
& \hspace{.5cm} - G_\bfm \cdot\nabla \left(F_{\Tilde{c}} + \nabla\cdot\left(\rho^{a}\lambda_s
\Gamma{\bfxi} F_{\si^a}\right)/\rho\right)\Big]\,.
\nn
\eal
 This bracket is clearly bilinear and skew-symmetric.  Because it was derived from the  bracket \ref{PB1} by a change of variables,  satisfaction of the Jacobi identity is assured. We note, as before,   strong boundary conditions are assumed such that all integrations by parts produce vanishing boundary terms.

Thus we have completed the third part of our algorithm, the construction of a Poisson bracket that has  entropy functional  of \eqref{CS2} in the set of its Casimir invariants.  {Recall  the integrand of the entropy is given by
\bq
\si_{\mathrm{Total}}^a\coloneqq  \si^a + \frac{\rho^{a} }{2} \lambda_{s} \Gamma^2(\nabla c) \,;
\label{sitotal}
\eq
so we find}
\bal
\frac{\de S^a}{\de \si^a}&=1\,,
\quad
   \frac{\de S^a}{\de \tilde{c}}= -  \frac1{\rho} \nabla\cdot\left( 
   \rho^{a} \lambda_{s} \Gamma {\bfxi}  \right)
\nn\\
\frac{\de S^a}{\de \rho}&= \frac{a}{2}\rho^{a-1}\la_s \Ga^2 + \frac{\Tilde{c}}{\rho^2}  \nabla\cdot \left(\rho^{a} \lambda_{s} \Gamma {\bfxi}\right)\,. 
\label{crule}
\eal
Using \eqref{crule}  one  can easily check that $\{F,S^a\}=0$ for all $F$, which by construction had to be the case. Now we are free to choose any Hamiltonian we desire in \eqref{GenBkt}   to obtain the evolution of any observable $o$ as follows: 
\bq
    \partial_t o = \{o, H^a\}^a\,,
    \label{fH2}
\eq
In \S\S~\ref{ssec:chnsHS} we proposed  the Hamiltonian functional  of  \eqref{TE1}, which we rewrite as follows in order to make all arguments clear:
\bal
 H^a[\rho,\mathbf{m},\sigma,\Tilde{c}]& = \int_{\Omega}  \frac{|\mathbf{m}|^2}{2\rho} + \rho u\left(\rho,\frac{\sigma^a}{\rho},\frac{\Tilde{c}}{\rho}\right)
 \nn\\
 & \hspace{1.6cm}+ \frac{\rho^{a}}{2}\lambda_u\Gamma^2\left(\nabla \left(\frac{\tilde{c}}{\rho}\right)\right) \,.
\eal
Using the functional derivatives of this Hamiltonian, 
\bal
   H^a_\bfm &= \mathbf{v},\quad H_{\si^a} = T\,,
   \nn\\
H^a_\rho &= - |\mathbf{v}|^2/2 + u +   {p}/{\rho} -sT - c\mu 
\nn\\
& \hspace{1cm}+  a \rho^{a -1}\lambda_u \Gamma^2 /2 + {c}\, \nabla\cdot (\rho^{a}\lambda_u\Gamma{\bfxi})/\rho \,,
\nn\\
  H^a_{\Tilde{c}} & = \mu - \nabla\cdot(\rho^{a}\lambda_u \Gamma{\bfxi})/\rho\,, 
    \label{Hover}
\eal
in the bracket \eqref{GenBkt} gives the equations of motion in the form of \eqref{fH2}.   At this point we could write this out and display a general system of equations that includes both cases, but we choose to consider them separately because the general system is unwieldy and not particularly  perspicuous. 

Let us first consider the  simplified version of our  derived Poisson bracket for the case $a= 1$, which is as follows:
\bal
\{F,G\}^1 &= -  \int_{\Omega} \mathbf{m}\,  \cdot \big[F_\bfm \cdot \nabla G_\bfm -G_\bfm \cdot\nabla F_\bfm \big]
\nn\\
&+ \rho \big[F_\bfm \cdot\nabla G_\rho - G_\bfm\cdot  \nabla F_\rho \big] 
\nn\\
&- \lambda_s\big[ F_\bfm  \cdot\nabla \cdot \left(\rho  G_{\si^1}  
\Gamma {\bfxi}\otimes\nabla c\right) 
\nn\\
&\hspace{2cm} - G_\bfm \cdot\nabla \cdot \left(\rho  F_{\sigma^1} \Gamma{\bfxi}\otimes\nabla c\right)\big]
\nn\\
&+ \sigma^1 \left[F_\bfm  \cdot\nabla G_{\si^1}  -  G_\bfm  \cdot\nabla F_{\si^1} \right]
\nn\\
&+ \Tilde{c} \left[F_\bfm  \cdot\nabla G_{\Tilde{c}} - G_\bfm \cdot\nabla F_{\Tilde{c}} \right]\,,
\label{question}
\eal
where $\otimes$ denotes tensor product of two vectors and consistent with our convention we have
\[
\nabla\cdot(\bfu\otimes\bfv)= (\nabla\cdot \bfu) \bfv  + \bfu\cdot\nabla \bfv\,.
\]
Using \eqref{Hover}, the bracket form of \eqref{fH2} gives the ideal diffuse two-phase flow system
\bal
\p_t \mathbf{v} &=\{\bfv,H^1\}^1
\nn\\
&= -\bfv\cdot\nabla \bfv -\frac{1}{\rho}\nabla \cdot (p\tensI  +\rho\lambda_f \Gamma{\bfxi}\otimes\nabla c)\,,
   \label{pvt1}
\\
\p_t \rho  &=\{\rho,H^1\}^1= -\bfv\cdot\nabla\rho   -\rho\, \nabla\cdot\mathbf{v}\,,
      \label{prt1}
\\
   \p_t \Tilde{c} &=\{\tilde{c},H^1\}^1= -\bfv\cdot\nabla\tilde{c} -\Tilde{c}\,\nabla\cdot\mathbf{v}\,, 
      \label{pct1}
\\
 \p_t \si^{1}_{\mathrm{Total}} &=\{{\si}^{1}_{\mathrm{Total}}  ,H^1\}^1
 \nn\\
 & = -\bfv\cdot\nabla{\si}^{1}_{\mathrm{Total}}  -{\si}^{1}_{\mathrm{Total}} \,\nabla\cdot\mathbf{v}\,,
      \label{pst1}
\eal 
where $\tensI$ is the unit tensor.   Observe in \eqref{pst1} we have chosen the observable $\si^{1}_{\mathrm{Total}}$
instead of $\si^1$, in order to demonstrate its conservation.

Similarly, for the case where $a = 0$, the Poisson bracket has the following form: 
\bal
\{F,G\}^0 &= -  \int_{\Omega} \mathbf{m}\cdot\, \big[F_\bfm\cdot\nabla G_\bfm -  G_\bfm  \cdot\nabla F_\bfm \big]
\nn \\
& + \rho\big[F_\bfm \cdot\nabla G_\rho  -G_\bfm \cdot\nabla F_\rho \big]
\nn \\
& -\lambda_s\big[F_\bfm \cdot\nabla \cdot\left(G_{\si^0}\Gamma{\bfxi}\otimes\nabla c\right) 
\nn\\
& \hspace{2.4cm} -G_\bfm \cdot\nabla \cdot \left(F_{\si^0} \Gamma{\bfxi}\otimes\nabla c\right)\big] 
\nn \\
&+ {\lambda_s} \left[F_\bfm \cdot\nabla\left( \Gamma^2    G_{\si^0}   \right) - G_\bfm
\cdot\nabla\left(\Gamma^2 F_{\si^0} \right)\right] /2
\nn \\
& + \sigma^0 \left[F_\bfm \cdot\nabla G_{\si^0} -  G_\bfm \cdot  \nabla   F_{\si^0} \right]
\nn\\
&+ \Tilde{c} \left[ F_\bfm \cdot\nabla   G_{\Tilde{c}} - G_\bfm  \cdot\nabla F_{\Tilde{c}}  \right]\,.
\eal
Same as above, using \eqref{Hover}, the ideal diffuse two-phase flow system is produced
\bal
  \p_t \mathbf{v}  &=\{\bfv,H^0\}^0
   \label{pvt0}\\
  &= -\bfv\cdot\nabla \bfv -\frac{1}{\rho}\nabla \cdot 
   \Big[\left(p -  \lambda_f\Gamma^2/2\right)\tensI 
   \nn\\
   &\hspace{3.25cm} +\lambda_f \Gamma{\bfxi}\otimes\nabla c\Big]\,,
\nn
\\
  \p_t \rho  &=\{\rho,H^0\}^0= -\bfv\cdot\nabla\rho   -\rho\, \nabla\cdot\mathbf{v}\,,
      \label{prt0}
\\
    \p_t \Tilde{c} &=\{\tilde{c},H^0\}^0= -\bfv\cdot\nabla\tilde{c} -\Tilde{c}\,\nabla\cdot\mathbf{v}\,,
      \label{pct0}
\\
 \p_t \si^{0}_{\mathrm{Total}} &=\{{\si}^{0}_{\mathrm{Total}}  ,H^0\}^0
 \nn\\
 & = -\bfv\cdot\nabla{\si}_{\mathrm{Total}} ^{0} -{\si}_{\mathrm{Total}} ^{0}\,\nabla\cdot\mathbf{v}\,.
      \label{pst0}
\eal 
where recall from \eqref{GF1}, $\xi = \partial\Ga/\partial \bfp$ .

Let us now comment on these two Hamiltonian systems.  By construction both the $a=1$ and $a=0$ systems conserve their Hamiltonians and entropies, as given   by \eqref{CS2} and \eqref{TE1} with $a=1$ and $a=0$ , respectively.  Both systems have  momentum equations containing  a term describing anisotropic surface energy (capillary) effects.  The $a=0$  system of \eqref{pvt0}--\eqref{pst0} is  identical to the ideal limit of that given  in  the work of \citet{anderson2000phase}.  Upon choosing $\Ga(\nabla c)= |\nabla c|$, the  $a=1$ system of  \eqref{pvt1}--\eqref{pst1}  should correspond to the ideal limit of that of \citet{Guo}, but it does not.  In fact the system of  \citet{Guo} in this limit does not conserve energy. Moreover,  the  capillary effect in their momentum equation (equation (3.40)), which  should be  replaced by  \eqref{pvt1}   with  $\Ga(\nabla c)= |\nabla c|$, vanishes  in the one-dimensional limit.  Since such surface effects are determined by  mean or weighted mean curvature \citep{taylor92}, it is clear that this is physically untenable.  Fortunately, our  method provides a simple fix to their equations, while showing how to generalize them to include anisotropic surface effects. 

An alternative but equivalent Hamiltonian formulation of the above systems exists, in fact, one that has  a standard entropy functional of the form of \eqref{entsi}.   Given that the bracket of \eqref{GenBkt} was obtained via a transformation of the bracket of \eqref{PB1}, we can transform it back from one that has \eqref{CS2} as a Casimir to the original  that has \eqref{entS} as a  Casimir.   However, to generate equivalent equations of motion, we would have to transform the Hamiltonian of \eqref{TE1} into a more complicated form.  Tracing back  through our transformations, we would replace the coordinate  $\si^a$ in the Hamiltonian by $\sigma -  {\rho^a} \lambda_{s} \Gamma^2 /2
$,   which means  the internal energy becomes
\bq
 u\mapsto u(\rho,(\si - {\rho^a} \lambda_{s} \Gamma^2 /2)/{\rho},\Tilde{c}/\rho)\,,
 \label{uNew}
 \eq
while otherwise the Hamiltonian remains the same.  Just as with finite-dimensional  Hamiltonian systems,  one can change coordinates and arrive at equivalent systems with different Poisson brackets and Hamiltonians, and in the noncanonical case different expressions for the Casimir invariants.  Often one has the options of a simple bracket and complicated Hamiltonian or vice verse.

\subsection{Metriplectic 4-bracket for the  Cahn-Hilliard-Navier-Stokes  system}
\label{ssec:chns4MB}

Now we turn to the 4th  and final step of our algorithm.  Just as in \S\S~\ref{sssec:GNS} we build a metriplectic 4-bracket using the K-N construction.  We suppose our multi-component field  variable  is
\bq
\psi(\bfx,t) = (\bfm(\bfx,t), \rho(\bfx,t),\Tilde{c}(\bfx,t), \sigma^{a}(\bfx,t))  
\label{psim}
\eq
and in the  K-N construction we use a more general expression for $\Si$ akin to that  mentioned in \eqref{Sgma2}, viz.
\bal
M(dF ,dG ) &= F_{{{\si}^{a}}}G_{{{\si}^{a}}}\,,
\\
{\Sigma}(dF ,dG ) &=  \nabla F_{{\bfm}} : \4tensLa_1: \nabla G_{{\bfm}} +  \nabla F_{{{\si}^{a}}}\cdot \tensLa_2\cdot \nabla G_{{{\si}^{a}}} 
\nn\\
&\hspace{1.5cm} + \nabla \mathcal{L}^a_{\Tilde{c}}(F) \cdot\tensLa_3\cdot \mathcal{L}^a_{\Tilde{c}}(F) \,,
\label{SiSI3}
\eal
where $a$, of course, is not to be summed over and the pseudodifferential operator $\mathcal{L}^a_{\Tilde{c}} $ has the following form: 
\bq
\mathcal{L}^a_{\Tilde{c}}(F) := \nabla \big( F_{\Tilde{c}}   + \nabla\cdot\left(\rho^{a}\lambda_s \Gamma{\bfxi} F_{\si^a} \right)/\rho\big)\,.
\eq
Here the tensors $\4tensLa_1$ , $\tensLa_2$ and $\tensLa_3$ are defined by  \eqref{TensorsLa}.  Then, the 4-bracket reads
\bal
(F, K ; G, N)^a&= 
\label{4BDICHNS}\\
& \hspace{-1.9cm} \int_{\Omega}\!\frac{1}{T}\Big[
\left[K_{{\si}^{a}} \nabla   F_{\mathbf{m}}-F_{{\si}^{a}} \nabla   K_{\mathbf{m}}\right] \!
\colon\!\! \4tensLa \colon\!\!
\left[N_{{\si}^{a}} \nabla   G_{\mathbf{m}}-G_{{\si}^{a}} \nabla   N_{\mathbf{m}}\right]
\nn\\
&\hspace{-1.88cm}+ \frac{1}{T}
\big[K_{{\si}^{a}} \nabla  F_{{{\si}^{a}}}-F_{{\si}^{a}} \nabla  K_{{{\si}^{a}}}\big]
\cdot\tenska\cdot \big[N_{{\si}^{a}} \nabla  G_{{{\si}^{a}}}-G_{{\si}^{a}} \nabla  N_{{{\si}^{a}}}\big]
\nn\\
&\hspace{-1.88cm}+ 
\!\big[K_{{\si}^{a}} \mathcal{L}^a_{\Tilde{c}}(F)  -F_{{\si}^{a}}  \mathcal{L}^a_{\Tilde{c}}(K)\big]\!\cdot\! \tensD\!\cdot\!
\big[N_{{\si}^{a}} \mathcal{L}^a_{\Tilde{c}}(G) - G_{{\si}^{a}}\mathcal{L}^a_{\Tilde{c}}(N)\big].
\nn
\eal
Observe, with the exception of the last line this bracket  is identical to that of \eqref{4BSICHNS}.
 
Upon insertion of  $H^a$ as given by \eqref{TE1}  and $S$  from the set of Casimirs to be as  in  \eqref{CS2}, the  dynamics is given by 
\bq
\p_t\psi^\al = \{\psi^\al, H^a\}^a + (\psi^\al,H^a;S^a,H^a)^a\,. 
\eq
Using  $\mathcal{L}^a_{\Tilde{c}}( {H^a}) =\nabla \mu^a_{\Ga}$,  $H^a_{\bfm} = \bfv$, $H^a_{{\si}_{a}} = T$, $S^a_{{\si}_{a}} = 1$ and    $\mathcal{L}^a_{\Tilde{c}}({S^a}) 
 = 0 $,
  the following diffuse-interface CHNS system for $a=1$ is produced:
 \bal
    \p_t \mathbf{v}  &=\{\bfv,H^1\}^1+ (\bfv,H^1;S^1,H^1)^1
   \nn\\
   &= -\bfv\cdot\nabla \bfv -\frac{1}{\rho}\nabla\cdot \big[ p\mathbf{I} +\lambda_f \rho\Gamma{\bfxi}\otimes\nabla c\big]
   \nn\\
   &\hspace{3 cm} + \frac1{\rho}\nabla\cdot(\4tensLa:\nabla \mathbf{v})\,,
   \label{Dpvt31}
\\
 \p_t \rho &=\{\rho,H^1\}^1 + (\rho,H^1;S^1,H^1)^1
 \nn\\
 &= -\bfv\cdot\nabla\rho   -\rho\, \nabla\cdot\mathbf{v} \,,
      \label{Dprt31}
      \\
 \p_t \Tilde{c}  &=\{\tilde{c},H^1\}^1+ (\tilde{c},H^1;S^1,H^1)^1
 \nn\\
 &= -\bfv\cdot\nabla\tilde{c} - \Tilde{c}\,\nabla\cdot\mathbf{v}
 + \nabla\cdot(\tensD\cdot\nabla\mu^1_\Ga)\,,
     \label{Dpct31}
\\
    \p_t \si^1_{\mathrm{Total}} &=\{\si^1_{\mathrm{Total}} ,H^1\}^1  + (\si^1_{\mathrm{Total}} ,H^1;S^1,H^1)^1
    \nn\\
    &= -\bfv\cdot\nabla  \si^1_{\mathrm{Total}} 
     - \si^1_{\mathrm{Total}}  \,\nabla\cdot\mathbf{v}
       \nn\\ &\quad + \nabla\cdot\left(\frac{\tenska}{T}\cdot\nabla T\right) + \frac{1}{T^2}\nabla T\cdot\tenska\cdot\nabla T  
          \label{Dpst31}\\
       &\hspace{1cm}+ \frac{1}{T}\nabla\mathbf{v}:\4tensLa: \nabla\mathbf{v} + \frac{1}{T} 
       \nabla\mu^1_\Ga \cdot \tensD\cdot\nabla\mu^1_\Ga \,.
 \nn
\eal 
Similarly,  for $a= 0$ we obtain
 \bal
  \p_t \mathbf{v} &=\{\bfv,H^0\}^0+ (\bfv,H^0;S^0,H^0)^0 
   \nn\\
   &= -\bfv\cdot\nabla \bfv -\frac{1}{\rho}\nabla \cdot \Big[\left(p -  \lambda_f\Gamma^2/2\right)\mathbf{I}
    \nn\\
    & \hspace{1cm}+\lambda_f \Gamma{\bfxi}\otimes\nabla c\Big]
 + \frac1{\rho}\nabla\cdot(\4tensLa:\nabla \mathbf{v})\,,
   \label{Dpvt30}
\\
    \p_t \rho  &=\{\rho,H^0\}^0+ (\rho,H^0;S^0,H^0)^0
    \nn\\
    &= -\bfv\cdot\nabla\rho   -\rho\, \nabla\cdot\mathbf{v} \,, 
      \label{Dprt30}
      \\
    \p_t \Tilde{c}  &=\{\tilde{c},H^0\}^0 + (\tilde{c},H^0;S^0,H^0)^0
     \nn\\
    &= -\bfv\cdot\nabla\tilde{c} -\Tilde{c}\,\nabla\cdot\mathbf{v}
 + \nabla\cdot(\tensD\cdot\nabla\mu^0_\Ga)\,,
      \label{Dpct30}
    \\
    \p_t \si^0_{\mathrm{Total}}  &=\{\si^0_{\mathrm{Total}} ,H^0\}^0 + (\si^0_{\mathrm{Total}} ,H^0;S^0,H^0)^0
      \nn\\
    &= -\bfv\cdot\nabla \si^0_{\mathrm{Total}} 
     -  \si^0_{\mathrm{Total}} \,\nabla\cdot\mathbf{v}
           \label{Dpst30}
           \\ &\quad + \nabla\cdot\left(\frac{\tenska}{T}\cdot\nabla T\right) + \frac{1}{T^2}\nabla T\cdot \tenska\cdot\nabla T  
            \nn\\
    &\hspace{1cm}+ \frac{1}{T}\nabla\mathbf{v}:\4tensLa : \nabla\mathbf{v} + \frac{1}{T} 
       \nabla\mu^0_\Ga\cdot \tensD \cdot \nabla\mu^0_\Ga \,.
  \nn
\eal 
Thus we have extracted from our general system with arbitrary $a$,   two  thermodynamically consistent  CHNS  systems.  By construction both the $a=1$ and $a=0$ systems must conserve energy   and both must produce entropy,  which we find is governed by the following: 
\bal
\dot{S^a}&= (S^a, H^a ; S^a, H^a)^a 
\nn\\
&= \int_\Om \frac{1}{T} \bigg[
\nabla\mathbf{v}:\4tensLa:\nabla\mathbf{v} + \frac{1}{T}\nabla T \cdot \tenska\cdot \nabla T 
\nn\\
&\hspace{3.10cm}+   \nabla\mu^a_\Ga\cdot \tensD\cdot\nabla\mu^a_\Ga
\bigg] \geq 0\,.
\label{dots3}
\eal

\subsection{Metriplectic 2-bracket for the  Cahn-Hilliard-Navier-Stokes  system}
\label{ssec:chns2MB}

Proceeding as  in the previous section, using $\mathcal{L}^a_{\Tilde{c}}( {H}^a) = \nabla(\mu -\frac1{\rho} \nabla\cdot(\la_f \rho^a\Ga \bfxi)) =\nabla  \mu_{\Ga}^a$,  $H^a_\bfm = \bfv$, $H^a_{{\si}^{a}} = T$, the metriplectic 2-bracket emerges directly from the 4-bracket as follows:
\bal
(F, {G})^a_{ {H^a}}&=(F,{H^a}; {G}, {H^a})^a
\label{4BDICHNSS}\\
&\hspace{-.6cm}=  \int_{\Omega}\frac{1}{T}\Big[
\left[T \nabla   F_{\mathbf{m}}-F_{{\si}^{a}} \nabla   \bfv \right]
: \4tensLa : 
\left[T \nabla   G_{\mathbf{m}}-G_{{\si}^{a}} \nabla   \bfv\right]
\nn\\
&\hspace{-.37cm}+\ \frac{1}{T}\,
\big[T \nabla  F_{{{\si}^{a}}}-F_{{\si}^{a}} \nabla  T\big]
\cdot \tenska\cdot\big[T \nabla  G_{{{\si}^{a}}}-G_{{\si}^{a}} \nabla  T\big]
\nn\\
&\hspace{-.37cm}+ 
\big[ T\mathcal{L}^a_{\tilde{c}}(F)-F_{\si^a}\nabla \mu_{\Ga}^a \big]
\!\cdot\tensD\cdot\!
    \big[T\mathcal{L}^a_{\tilde{c}}(G)-G_{\si^a}\nabla \mu_{\Ga}^a \big]  
\Big]\,,
\nn
\eal
which is the analog of \eqref{4BSICHNS} for the GNS system in \S\S~\ref{sssec:GNS}.

Now we could proceed  as in  \S\S~\ref{sssec:genTh} and obtain the analog of $L_{\al\be}$ by transforming to  the  coordinates  $\xi^a \!\! \coloneqq (\bfm , \rho ,\Tilde{c},e^a_{\mathrm{Total}})$, 
where $e^a_{\mathrm{Total}}$ is the total energy density defined in \eqref{TE1}.  This would lead to a metriplectic 2-bracket  analogous to \eqref{2bktxi} and the concomitant   flux and affinity relations analogous to \eqref{tensorL2} would emerge.  Because this is hardly more enlightening than \eqref{4BDICHNSS}, we do not record this result here.  

\section{Summary and Conclusions}
\label{sec:conclu}

In this paper we have described the metriplectic 4-bracket formalism and how it can algorithmically be used in the context of multiphase fluids to construct thermodynamically consistent models, ones that conserve energy and produce entropy.  In particular, we have used it in \S~\ref{sec:GNS} to obtain the GE Hamiltonian system,  which adds a concentration variable to conventional Eulerian fluid mechanics, and to the GNS system,  a generalization of the Navier-Stokes system that is thermodynamically consistent with the collection of  thermodynamic fluxes. Then, in \S~\ref{sec:DICHNS} we used the algorithm to obtain a class of Hamiltonian fluid systems that allow for anisotropic surface effects, followed by the construction of   a general class of CHNS systems that  couple Cahn-Hilliard physics with that of  Navier-Stokes dynamics,  in a thermodynamically consistent way.  The systems we obtain generalizes previous work by including anisotropic effects in the surface tension and all phenomenological parameters.   

A cornerstone of Hamiltonian dynamics is its geometric invariance under coordinate changes.  Because the minimal metriplectic properties are algebraic and geometric, they too are invariant under coordinate changes.  Thus,  we can write our  CHNS class of dissipative systems with a standard entropy functional of the form of \eqref{entsi}, but with a more complicated Hamiltonian using \eqref{uNew}.

From the examples presented, it is clear that  the  4-bracket formalism can be applied to obtain a wide variety of dynamical systems in various fields.  In fact it was recently applied to obtain generalized collision operators in kinetic theory \citep{pjmS24} and a thermodynamically consistent model for radiation hydrodynamics \citep{tran_etal}. Although incompressible flows don’t have the usual thermodynamics associated with compression and pressure, they can be included in the metriplectic formalism by using the techniques of Chandre et al. (2013).

The metriplectic 4-bracket formalism also provides an avenue for designing structure preserving numerical algorithms 
\citep[see e.g.][]{pjm17}.  Any discretization that preserves the symmetries of the 4-bracket, which is not a difficult task,  will be thermodynamically consistent on the semi-discrete level, i.e. produce a set of ordinary differential equations that conserve energy and produce entropy.

\hspace{1cm}

\appendix


\section{Semidefinite curvature}
 \label{A:seccurve}
 
We define the binary operations $\langle \cdot, \cdot\rangle_\Si$ and $\langle \cdot, \cdot\rangle_M$ that satisfy all of the axioms of an inner
product space, except the non-degeneracy condition 
\bal
\langle F, G\rangle_\Si&:=\int  \! \!  d^N\!z \!\! \int \! \! d^N\!z'\, \Si^{\al\be}(z,z') \frac{\delta F}{\delta \chi^\al(z)} \frac{\delta G}{\delta \chi^\be(z')}\,,
\nn\\
\langle F, G\rangle_M&:=\!\!\int  \! \!  d^N\!z'' \!\! \int \! \! d^N\!z'''\, M^{\ga\de}(z'',z''') \frac{\delta F}{\delta \chi^\ga(z'')} \frac{\delta G}{\delta \chi^\de(z''')}\,,
\nn
\eal
where $\Si$ and $M$ are positive semi-definites. We have the Cauchy-Schwarz inequality
\bq
\left|\langle F, G\rangle_\Si\right| \leq \sqrt{\langle F, F\rangle_\Si} \,\,\sqrt{\langle G, G\rangle_\Si} 
=\|F\|_\Si\|G\|_\Si\,.
\nn
\eq

\begin{lemma}
A metriplectic quadravector constructed using the $K-N$ product, has non-negative sectional curvature,
\bal
K(F,G) &= \!\!\!\int  \! \!  d^N\!z \!\! \int \! \! d^N\!z'\,\! \!\! \int  \! \!  d^N\!z'' \!\!\! \int \! \! d^N\!z'''\, \Si^{ij}(z,z')\, M^{kl}(z''\!\!,z''')
\nn\\
&\hspace{-.75cm}  \times \  \frac{\delta F}{\delta \chi^\al(z)} \frac{\delta G}{\delta \chi^\be(z')}\,
 \frac{\delta F}{\delta \chi^\ga(z'')} \frac{\delta G}{\delta \chi^\de(z''')}
 \	+ \mathrm{\ other\ terms}
 \,.\nn
\eal
\end{lemma} 
\begin{proof}
Direct calculation gives
\bq
K(F, G)=\|F\|_\Si^2\|G\|_M^2-2\langle F, G\rangle_\Si\langle F, G\rangle_M+\|G\|_\Si^2\|F\|_M^2 
\nn\,.
\eq
The following inequality
\bq
\left(\|F\|_\Si\|G\|_M-\|G\|_\Si\|F\|_M\right)^2 \geq 0 
\nn
\eq
implies
\bal
\|F\|_\Si^2\|G\|_M^2+\|G\|_\Si^2\|F\|_M^2 &\geq 2\|F\|_M\|F\|_\Si\|G\|_M\|G\|_\Si \nn\\
&\geq 2\left|\langle F, G\rangle_\Si\right|\left|\langle F, G\rangle_M\right| \nn \\
& \geq 2\langle F, G\rangle_\Si\langle F, G\rangle_M \nn\,,
\eal
where the second inequality follows from  the Cauchy-Schwarz inequality. Evidently, the last inequality implies $K(F, G) \geq 0$ for all $F$ and $G$.
\end{proof} 
\begin{lemma}
We suppose that $\Si$ is positive definite, defining an
inner product. Given any two $\Si$-arbitrary linearly independent $\de F/\de{\chi}$ and $\de G/\de{\chi}$, then the
sectional curvature is strictly positive ($K(F,G) > 0$).
\end{lemma}

\begin{proof}
Since $\de F/ \de{\chi}$ and $\de G/\de {\chi}$ are $\Si$-Linearly independent, the Cauchy-Schwarz inequality given by 
\bq
\left|\langle F, G\rangle_\Si\right| < \|F\|_\Si\|G\|_\Si\,.
\nn
\eq
In the same way we have 
 \bq
\left(\|F\|_\Si\|G\|_M-\|G\|_\Si\|F\|_M\right)^2 \geq 0 \nn
\eq
implies
\bal
\|F\|_\Si^2\|G\|_M^2+\|G\|_\Si^2\|F\|_M^2 &\geq 2\|F\|_M\|F\|_\Si\|G\|_M\|G\|_\Si \nn \\
& > 2\left|\langle F, G\rangle_\Si\right|\left|\langle F, G\rangle_M\right| \nn\\ 
& > 2\langle F, G\rangle_\Si\langle F, G\rangle_M\,.
\nn
\eal
Hence, we deduce that $K(F,G) > 0$\,.
\end{proof}

 {Finite-dimensional versions of these two lemmas were first reported in \citet{pjmU23}.}

 {\section*{Acknowledgements}

A.Z. acknowledges support from the Mohammed VI Polytechnic University for supporting an internship at the University of Texas at Austin.   P.J.M. acknowledges support from the DOE Office of Fusion Energy Sciences under DE-FG02-04ER-54742, and from a Forschungspreis from the Alexander von Humboldt Foundation. }

\bibliographystyle{elsarticle-harv}\biboptions{authoryear}

\end{document}